\documentclass[sigconf,screen]{acmart}

\usepackage{algorithm}
\usepackage[italicComments=false]{algpseudocodex}
\usepackage{amsmath}
\numberwithin{equation}{section}
\usepackage{interval}
\usepackage{relsize}
\usepackage{graphicx}
\usepackage{hyperref}
\usepackage{enumitem}
\usepackage{xspace}
\usepackage{makecell}
\usepackage{tablefootnote}
\usepackage{tabularray}
\usepackage{microtype}
\UseTblrLibrary{booktabs}

\usepackage[capitalise, noabbrev]{cleveref}
\crefformat{section}{\S#2#1#3}
\Crefformat{section}{\S#2#1#3}

\usepackage{threeparttable}

\usepackage{tikz}
\newcommand*\circled[1]{\tikz[baseline=(char.base)]{
            \node[shape=circle,fill,inner sep=1pt] (char) {\textcolor{white}{{\footnotesize #1}}};}}

\usepackage{xcolor}

\ifx\ANNOTATION\undefined
    \newcommand{\changescolor}{black}
\else
    \newcommand{\changescolor}{blue}
\fi

\let\oldcomment=\comment
\let\comment=\undefined
\usepackage[defaultcolor={\changescolor}, authormarkup=none]{changes}
\setcommentmarkup{ \revisionComment{#1}}

\let\comment=\oldcomment

\ifx\ANNOTATION\undefined
    \usepackage{xparse}
    \NewDocumentCommand{\noopadded}{o m}{#2}
    \NewDocumentCommand{\noopdeleted}{o m}{#2}
    \NewDocumentCommand{\noopreplaced}{o m m}{#2}
    \renewcommand\added\noopadded
    \renewcommand\deleted\noop
    \renewcommand\replaced\noopreplaced
    \renewcommand\highlight\noopadded
\fi

\NewCommandCopy{\oldBigcup}{\bigcup}
\renewcommand{\bigcup}{\mathsmaller\oldBigcup\limits}
\newcommand{\TextOrMathSimple}[1]{\TextOrMath{$#1$}{#1}}
\newcommand{\smallTag}[1]{\TextOrMathSimple{\mathsmaller{\mathrm{\textsc{#1}}}}}
\newcommand{\algoSpacing}{\vspace{2pt}}
\newcommand{\algoSize}{\small}
\renewcommand{\gets}{\mathrel{{:}{=}}}
\setlist[itemize]{leftmargin=*}
\renewcommand\proofname{Discussion}

\newcommand{\parabf}[1]{\noindent\textbf{#1}}
\newcommand{\revisionComment}[1]{\textcolor{gray}{\textit{[#1]}}}
\newcommand\COMMENT[1]{}
\newcommand{\CodeIn}[1]{{\small \texttt{#1}}}

\newcommand{\toolname}{\textrm{Repilot}}
\newcommand{\tool}{\toolname\xspace}
\newcommand{\ide}{IDE\xspace}
\newcommand{\idefull}{Integrated Development Environment\xspace}
\newcommand{\caret}{caret\xspace}

\newcommand{\cefull}{Completion Engine\xspace}

\newcommand{\plm}{LLM\xspace} 
\newcommand{\plmfull}{Large Language Model\xspace}
\newcommand{\nlp}{NLP\xspace}
\newcommand{\nlpfull}{Natural Language Processing\xspace}
\newcommand{\llm}{LLM\xspace} %
\newcommand{\llmfull}{Large Language Model\xspace}
\newcommand{\coreApproachDef}{Completion-Guided Search Space Pruning\xspace}
\newcommand{\mlmfull}{Masked Language Modeling\xspace}
\newcommand{\mlm}{MLM\xspace}
\newcommand{\mspfull}{Masked Span Prediction\xspace}
\newcommand{\msp}{MSP\xspace}
\newcommand{\apr}{APR\xspace}
\newcommand{\aprfull}{Automated Program Repair\xspace}
\newcommand{\nmt}{NMT\xspace}
\newcommand{\nmtfull}{Neural Machine Translation\xspace}
\newcommand{\codebert}{CodeBERT\xspace}
\newcommand{\graphcodebert}{GraphCodeBERT\xspace}
\newcommand{\codegen}{CodeGEN\xspace}
\newcommand{\codex}{Codex\xspace}
\newcommand{\polycoder}{PolyCoder\xspace}
\newcommand{\codetf}{CodeT5\xspace}
\newcommand{\incoder}{\textsc{InCoder}\xspace}
\newcommand{\plbart}{PLBART\xspace}

\newcommand{\abvanilla}{$\tool_\varnothing$}
\newcommand{\abnomem}{$\tool_{\textsc{p}}$}
\newcommand{\abmem}{$\tool_{\textsc{p}}^{\textsc{m}}$}
\newcommand{\abactive}{\tool}

\newcommand\chatrepair{\textsc{ChatRepair}\xspace}
\newcommand\fitrepair{FitRepair\xspace}
\newcommand{\rewardrepair}{RewardRepair\xspace}
\newcommand{\alpharepair}{AlphaRepair\xspace}
\newcommand{\cure}{CURE\xspace}
\newcommand{\recoder}{Recoder\xspace}
\newcommand{\coconut}{CoCoNuT\xspace}
\newcommand{\dlfix}{DLFix\xspace}
\newcommand{\sequencer}{SequenceR\xspace}
\newcommand{\tbar}{TBar\xspace}
\newcommand{\prapr}{PraPR\xspace}
\newcommand{\avatar}{AVATAR\xspace}
\newcommand{\simfix}{SimFix\xspace}
\newcommand{\fixminer}{FixMiner\xspace}
\newcommand{\capgen}{CapGen\xspace}
\newcommand{\jaid}{JAID\xspace}
\newcommand{\sketchfix}{SketchFix\xspace}
\newcommand{\nopol}{NOPOL\xspace}
\newcommand{\jgenprog}{jGenProg\xspace}
\newcommand{\jmutrepair}{jMutRepair\xspace}
\newcommand{\jkali}{jKali\xspace}
\newcommand{\iclr}{\textsc{Synchromesh}\xspace}

\newcommand{\dfj}{Defects4J\xspace}

\newcommand{\fmVar}[1]{\TextOrMathSimple{\mathit{#1}}}
\newcommand{\fmType}[1]{\TextOrMathSimple{\mathsf{#1}}}
\newcommand{\fmConst}[1]{\TextOrMathSimple{\mathbf{#1}}}
\newcommand{\fmSeq}[1]{\TextOrMathSimple{#1^*}}
\newcommand{\fmPow}[1]{\TextOrMathSimple{\mathcal{P}(#1)}}
\mathchardef\fmHyphen="2D

\newcommand{\fmStaticChecker}{\TextOrMathSimple{\Phi}}

\newcommand{\fmProgWithCheckRhs}{\TextOrMathSimple{(\fmCeAlphabet, \fmStaticChecker)}}
\newcommand{\fmProgWithCheck}{\TextOrMathSimple{\fmDef{PL}_\smallTag{s}}}
\newcommand{\fmProgWithCheckEq}{\TextOrMathSimple{\fmProgWithCheck{} = \fmProgWithCheckRhs{}}}
\newcommand{\fmFeasibleMark}{\vDash}
\newcommand{\fmFeasible}[1]{\TextOrMathSimple{#1\fmFeasibleMark{}\fmStaticChecker{}}}
\newcommand{\fmPrefixSet}{\TextOrMathSimple{\mathrm{Prefix}}}
\newcommand{\fmInfeasible}[1]{\TextOrMathSimple{#1\not\fmFeasibleMark{}\fmStaticChecker{}}}

\newcommand{\fmCont}{\fmVar{cont}}

\newcommand{\fmDef}[1]{\TextOrMathSimple{\mathbf{#1}}}

\newcommand{\fmPlm}{\fmDef{LM}}
\newcommand{\fmPlmTok}{\TextOrMathSimple{\Sigma_{\smallTag{lm}}}}

\newcommand{\fmEnc}{\TextOrMathSimple{\fmVar{encode}}}
\newcommand{\fmDec}{\TextOrMathSimple{\fmVar{decode}}}
\newcommand{\fmEncRep}{\TextOrMathSimple{\fmType{EncRep}}}
\newcommand{\fmDecRep}{\TextOrMathSimple{\fmType{DecRep}}}

\newcommand{\fmDecRepRhs}{\TextOrMathSimple{\fmPlmTok \to \interval 0 1}}
\newcommand{\fmDecRepDef}{\TextOrMathSimple{\fmDecRep = \fmPlmTok \to \interval 0 1}}

\newcommand{\fmCeAlphabet}{\TextOrMathSimple{\Sigma_{\smallTag{pl}}}}
\newcommand{\fmCe}{\fmDef{CE}}

\newcommand{\fmCeNone}{\fmConst{unknown}}
\newcommand{\fmCeComplete}{\fmVar{complete}}

\newcommand{\fmCeProg}{\fmVar{prog}}
\newcommand{\fmCeCaret}{\fmVar{\caret}}
\newcommand{\fmOfType}{{\:\!{:}\:}}
\newcommand{\fmCeCompRet}{\fmPow{\fmSeq{\fmCeAlphabet}}}
\algnewcommand\fmAlgoInputsValue{\textbf{Inputs:}}
\algnewcommand\fmAlgoInputs{\item[\fmAlgoInputsValue]}
\algnewcommand\fmAlgoOutputValue{\textbf{Output:}}
\algnewcommand\fmAlgoOutput{\item[\fmAlgoOutputValue]}
\algrenewcommand\algorithmicfunction{\textbf{func}}
\newcommand{\algoFuncColon}{\TextOrMathSimple{{\:\!{:}}}}
\algdef{SxN}[FUNC]{Func}{EndFunc}[3]{%
    \algorithmicfunction\ \textproc{#1}\ifstrempty{#2}{}{(#2)}%
    \ $\to$ #3\algoFuncColon{}%
}{}
\algdef{SxN}[FUNC-NO-RET]{FuncNoRet}{EndFuncNoRet}[3]{%
    \algorithmicfunction\ \textproc{#1}\ifstrempty{#2}{}{(#2)}%
}{}

\newcommand{\fmFunc}[1]{\textsf{\textbf{\textsc{#1}}}}
\newcommand{\fmUtilFunc}[1]{\textsc{#1}}
\newcommand{\fmRepair}{\fmFunc{Repair}}
\newcommand{\fmActiveComplete}{\fmFunc{Actively\-Complete}}
\newcommand{\fmPruneDecode}{\fmFunc{Guid\-ed\-Prune}}
\newcommand{\fmEndToken}{\fmConst{end\fmHyphen token}}
\newcommand{\fmProg}{\fmVar{prog}}
\newcommand{\fmRange}{\fmVar{range}}
\newcommand{\fmRangeStart}{\fmVar{start}}

\newcommand{\fmGen}{\fmVar{hunk}}

\newcommand{\fmNextToken}{\fmVar{next\fmHyphen token}}
\newcommand{\fmPlmInputs}{\fmVar{encoder\fmHyphen inputs}}
\newcommand{\fmBuildInputs}{\fmUtilFunc{BuildInputs}}
\newcommand{\fmConcat}{\fmUtilFunc{Str}}
\newcommand{\fmSample}{\fmUtilFunc{Sample}}
\newcommand{\fmProgType}{\fmSeq{\fmCeAlphabet}}
\newcommand{\fmDecodeDetail}{\fmVar{decoder}}
\newcommand{\fmPatch}{\fmVar{patch}}
\newcommand{\fmTokens}{\fmVar{tokens}}
\newcommand{\fmToken}{\fmVar{tok}}
\newcommand{\fmRejected}{\fmVar{rejected}}
\newcommand{\fmAccepted}{\fmVar{accepted}}
\newcommand{\fmCompVar}{\fmVar{c}}
\newcommand{\fmContinuation}{\fmVar{completion}}
\newcommand{\fmCompletions}{\fmVar{completions}}
\newcommand{\fmCompletionTokens}{\fmVar{completion\fmHyphen toks}}
\newcommand{\fmCommonPrefix}{\fmUtilFunc{CommonPrefix}}
\newcommand{\fmPlmEncRep}{\fmVar{encoded\fmHyphen rep}}
\newcommand{\fmTokenDecode}{\fmUtilFunc{AlignTokens}}

\newcommand{\fmKwTrue}{\fmConst{true}}

\newcommand{\fmKwContinue}{\fmConst{continue}}

\newcommand{\fmKwAnd}{\fmConst{and}}

\newcommand{\fmArbFunc}{\fmVar{f}}
\newcommand{\fmArbPartialFunc}{\fmVar{a}}
\newcommand{\fmArbDom}{\fmType{X}}
\newcommand{\fmArbRange}{\fmType{Y}}
\newcommand{\fmArbFuncType}{\TextOrMathSimple{\fmArbDom \to \fmArbRange}}
\newcommand{\fmArbPartialFuncType}{\TextOrMathSimple{\fmArbDom \rightharpoonup \fmArbRange}}
\newcommand{\fmArbIn}{\fmVar{x}}
\newcommand{\fmArbOut}{\fmVar{y}}
\newcommand{\fmRemoved}{\TextOrMathSimple{\fmArbFunc_{\mathrm{removed}}}}

\newcommand{\distance}{6pt}
\setcopyright{acmlicensed}
\acmPrice{15.00}
\acmDOI{10.1145/3611643.3616271}
\acmYear{2023}
\copyrightyear{2023}
\acmSubmissionID{fse23main-p226-p}
\acmISBN{979-8-4007-0327-0/23/12}
\acmConference[ESEC/FSE '23]{Proceedings of the 31st ACM Joint European Software Engineering Conference and Symposium on the Foundations of Software Engineering}{December 3--9, 2023}{San Francisco, CA, USA}
\acmBooktitle{Proceedings of the 31st ACM Joint European Software Engineering Conference and Symposium on the Foundations of Software Engineering (ESEC/FSE '23), December 3--9, 2023, San Francisco, CA, USA}
\received{2023-02-02}
\received[accepted]{2023-07-27}

\begin{document}
\title{Copiloting the Copilots: Fusing Large Language Models with Completion Engines for Automated Program Repair}

\author{Yuxiang Wei}
\affiliation{
  \institution{University of Illinois}
  \city{Urbana-Champaign}
  \country{USA}                    %
}
\email{ywei40@illinois.edu}          %

\author{Chunqiu Steven Xia}
\affiliation{
  \institution{University of Illinois}
  \city{Urbana-Champaign}            %
  \country{USA}                    %
}
\email{chunqiu2@illinois.edu}          %

\author{Lingming Zhang}
\affiliation{
  \institution{University of Illinois}
  \city{Urbana-Champaign}
  \country{USA}                    %
}
\email{lingming@illinois.edu}          %

\begin{abstract}

During Automated Program Repair (APR), it can be challenging to synthesize correct patches for real-world systems in general-purpose programming languages.
Recent \plmfull{s} (\plm{s}) have been shown to be helpful ``copilots'' in assisting developers with various coding tasks, and have also been
directly applied for patch synthesis. However, most \plm{s} treat programs as sequences of tokens, meaning that they are ignorant of the underlying semantics constraints of the target programming language. This results in plenty of statically invalid generated patches, impeding the practicality of the technique. 
Therefore, we propose $\mathbf{\toolname}$, a general code generation framework to further copilot the AI ``copilots'' (i.e., \plm{s}) by synthesizing more \emph{valid} patches during the repair process.
Our key insight is that many \plm{s} produce outputs autoregressively (i.e., token by token), resembling human writing programs, which can be significantly boosted and guided through a \cefull.
\tool{} synergistically synthesizes a candidate patch through the interaction between an \plm{} and a \cefull, which 1) prunes away infeasible tokens suggested by the \plm and 2) proactively completes the token based on the suggestions provided by the \cefull.
Our evaluation on a subset of the widely-used Defects4j 1.2 and 2.0 datasets shows that \tool{} outperforms state-of-the-art techniques by fixing 27\% and 47\% more bugs, respectively. Moreover, \tool{} produces more valid and correct patches than the base LLM with the same budget. 
While we focus on leveraging \tool for \apr in this work, the overall approach is also generalizable to other code generation tasks.
\end{abstract}



\maketitle

\section{Introduction}

\begin{figure*}
\centering
\includegraphics[width=0.8\linewidth, page=4]{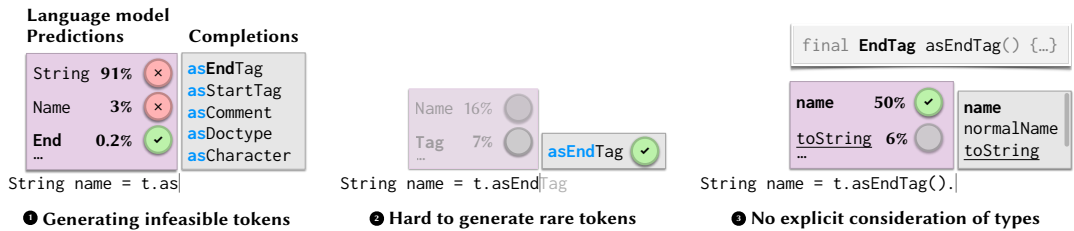}
\caption{Limitations of existing \llm-based \apr approaches.}
\label{fig:example}
\end{figure*}

\aprfull (\apr) seeks to reduce the manual bug-fixing effort of developers by automatically synthesizing patches given the original buggy code~\cite{gazzola2019aprsurvey}.
State-of-the-art traditional \apr tools are mainly based on handcrafted repair templates to match the buggy code patterns and apply the corresponding code changes~\cite{ghanbari2019prapr, liu2019tbar}. Although outperforming other traditional techniques~\cite{long2015spr, legoues2012genprog, mechtaev2016angelix}, such tools can only fix the bug types within the preset templates and cannot generalize to new bug types. 
With the development of Deep Learning (DL) techniques, researchers build learning-based \apr~\cite{zhu2021recoder, jiang2021cure, ye2022rewardrepair} tools based on \nmtfull (\nmt)~\cite{sutskever2014mt} architecture. They train \nmt models to translate buggy code into correct code by learning from pairs of buggy and fixed code scraped from open-source commits. However,
as discussed in prior work~\cite{xia2022alpharepair}, the training sets of these tools can be limited in size and also contain irrelevant or noisy commits.

More recently, researchers have leveraged the growth in the field of \nlp to directly use \llmfull{s} (\llm{s})~\cite{chen2021codex,feng2020codebert}
for \apr~\cite{xia2022repairstudy,joshi2023repair, xia2022alpharepair}. \llm{s} not only achieve impressive performance on many \nlp tasks~\cite{brown2020gpt3}, but are also shown to be reliable ``copilots''\footnote{One popular AI pair programmer tool (based on Codex~\cite{chen2021codex}) is named Copilot~\cite{GithubCopilot}.} in assisting developers with various coding tasks~\cite{austin2021synthesis, liu2023your}. The reason is that modern \llm{s} often include large amounts of available open-source code repositories as part of their training dataset. 
Recognizing the power of \llm{s}, researchers have recently applied \llm{s} for \apr: instead of translating buggy code into correct code, \llm{s} are directly used to synthesize the correct patch from the surrounding context. 
\alpharepair~\cite{xia2022alpharepair} reformulates the \apr problem as a \textit{cloze} (or infilling) task~\cite{aghajanyan2022cm3,fried2023incoder}: it first replaces the buggy code snippets with masked tokens and then uses \codebert~\cite{feng2020codebert} to fill correct code in given the surrounding context. Other studies on \llm{s} for \apr have applied even larger \llm{s} with different repair settings (including generating complete patch functions)~\cite{xia2022repairstudy, kolak2022patch, prenner2021codexws}. 

While prior \llm{} for \apr techniques achieve state-of-the-art bug-fixing performance, they use \llm{s} in a black-box manner, where the underlying \llm{} generate programs according to the token distribution without any structural or semantic understanding of the code.
To highlight the limitations with current \llm{s} for \apr tools, In \Cref{fig:example} we show 3 scenarios where \llm can generate incorrect patches. \textbf{\circled{1} Generating infeasible tokens}. In \Cref{fig:example}.1, the \llm has a high probability (>90\%) of generating \CodeIn{String} to complete the \CodeIn{asString} method. However \CodeIn{asString} is not a valid field access for the object \CodeIn{t} and is also not part of the scope of the current buggy method. In this case, the patchs generated using \CodeIn{asString} will never be correct as it cannot compile. By directly using the model probabilities, \llm{s} are likely to generate many patches using invalid tokens and decrease the likelihood of generating the correct patch with \CodeIn{End} (0.2\%). \textbf{\circled{2} Hard to generate rare tokens}. \llm{s} usually cannot generate a complete identifier name in one step since it uses subword tokenization~\cite{sennrich2015neural} to break uncommon words into smaller subwords. These uncommon words manifest as rare identifiers in code, where identifier names are CamelCase or underscore combinations of multiple words (e.g., \CodeIn{asEndTag} in \Cref{fig:example}.2). As such, \llm{s} need to generate these identifiers step by step, needing not only multiple iterations but also accurate output in each step. Since prior approaches~\cite{xia2022repairstudy, kolak2022patch} sample based on probability, the likelihood of completing a rare token to fix a bug can be extremely low. \textbf{\circled{3} No explicit consideration of types}. In addition to potentially generating out-of-scope identifiers, \llm{s} do not have access to various type information that gives hints to the valid identifiers. In \Cref{fig:example}.3, the return type of \CodeIn{asEndTag()} is \CodeIn{EndTag}, whose definition is not explicitly given to the \llm in its immediate context. As such, \llm{s} do not know the correct member fields of \CodeIn{EndTag} and may generate invalid patches containing identifiers that do not fit the required type. On the contrary, a \cefull has full access to the project and can easily figure out the return type of \CodeIn{asEndTag()} through static analysis on the abstract syntax tree of the program. By treating code as a sequence of textual tokens, the important type information is not encoded.

To address the aforementioned limitations, we propose \tool{}, a framework to further copilot the AI ``copilots'' (i.e., \llm{s}) via fusing \llm{s} with \cefull{s} to synthesize more \textit{valid} patches. \cefull{s}~\cite{lsp} can parse incomplete programs and reason about the semantics in an error-tolerant manner. Our key insight is \emph{to liken \llm autoregressive token generation as a human developer code writing, where the \cefull can provide real-time updates to check if the human/\llm{s} written partial code is valid}. \tool{} first uses the \llm to provide the probabilities of generating the next token in the patch and then queries the \cefull to modify the probability list by dynamically zeroing the probabilities of invalid tokens. We can then sample from the new probability list to select the next token. Furthermore, recognizing the ability for \cefull{s} to suggest completions, we use this feature whenever there is only one possible identifier suffix to complete the context. This not only allows \tool to generate patches with valid rare and long identifiers but also reduces the work of \llm{s} needed to iteratively generate long identifier names.

For example, \tool{} directly prunes the \CodeIn{String} and \CodeIn{Name} tokens in \Cref{fig:example}.1 as they are infeasible according to the \cefull{}, but still accepts the correct \CodeIn{End} token. In \Cref{fig:example}.2, the \cefull recognizes that \CodeIn{asEndTag} is the only valid continuation to the prefix \CodeIn{asEnd}, so \tool{} directly completes this token without querying the \llm{}. To combat the time cost of \cefull, we implement several optimization techniques to minimize the overhead. Note that the recent \iclr work~\cite{poesia2022synchromesh} also employs a \cefull for reliable code generation with \llm{s}. However, it relies on expert-designed constraints and only targets domain-specific languages (e.g., SQL). \tool{} directly works for general-purpose programming languages while introducing minimal overhead and can proactively complete the current generation using the \cefull{} without querying the \llm{}.

To demonstrate the generalizability of \tool, we instantiate \tool with two \llm{s} having distinct architectures and sizes:
\codetf{}-large~\cite{wang2021codet5}, an encoder-decoder \llm with 770 million parameters, and \incoder-6.7B~\cite{fried2023incoder}, a decoder only \llm with 6.7 billion parameters, both capable of code infilling from prefix and suffix context.
We further implement a Java \cefull for \tool based on the Eclipse JDT Language Server~\cite{eclipseLS,modifiedJDT} since it provides various semantics-based analyses through a consistent Language Server Protocol~\cite{lsp}.
We evaluate \tool on a subset of the widely studied \dfj 1.2 and 2.0 datasets~\cite{just2014dfj} and demonstrate state-of-the-art results in both the number of correct fixes and compilation rate --- the percentage of the generated patches that can be successfully compiled. Furthermore, while we evaluated \tool for \apr in this work, we believe the overall framework can be easily applied to other code generation tasks, including code completion~\cite{ding2023crosscodeeval, zhang2023repocoder}, program synthesis~\cite{liu2023your,poesia2022synchromesh}, and test generation~\cite{titanfuzz, fuzz4all}. In summary, we make the following contributions: 
\begin{itemize}
    \item \textbf{Direction.} We open a new direction for fusing \llm{s} with \cefull{s} for more powerful \apr and beyond. Compared to prior techniques which either perform post-processing to fix invalid generations or use simple static methods to approximate these valid tokens, our approach leverages a powerful \cefull to directly provide accurate feedback on partial programs to avoid invalid token generations.
    \item \textbf{Technique.} We implement \tool, an \llm for \apr approach instantiated with the \codetf and \incoder models to perform cloze-style repair combined with our modified Eclipse JDT Language Server~\cite{eclipseLS,modifiedJDT} as the \cefull{}.
    In \tool, we use the \cefull to systematically prune invalid tokens generated by \llm{s} and to directly complete code given the current prefix. Furthermore, we implement optimizations to significantly reduce the overhead of \tool. We have open-sourced our tool at: \url{https://github.com/ise-uiuc/Repilot}.
    \item \textbf{Study.} We compare \tool against state-of-the-art \apr tools on \dfj 1.2 and 2.0. \tool is able to achieve new state-of-the-art results of 66 \dfj 1.2 single-hunk bugs and 50 \dfj 2.0 single-line bugs fixed respectively with 30 more combined fixes across both datasets compared to the previous best baseline. Our further evaluation shows that \tool consistently improves the validity and correctness of the generated patches with a limited overhead (7\% for \codetf and negligible for \incoder).
\end{itemize}

\section{Background and related work}

\subsection{\llmfull{s} for Code}

Recent advances in \nlpfull (\nlp) have empowered the idea of using \plmfull{s} (\plm{s}) that are pre-trained on enormous corpora of natural language and code for various code-related tasks~\cite{austin2021synthesis, chen2021codex, xu2022systematic, groundcopilot, alphacode}.
\plm{s} are based on the transformer architecture~\cite{vaswani2017attention} that can be categorized into \textbf{encoder-only}, \textbf{decoder-only} and \textbf{encoder-decoder}. Encoder-only models use only the encoder component by training using \mlmfull (\mlm)~\cite{devlin2018bert} objective where a small percentage (e.g., 15\%) of the tokens are masked on. The goal of \mlm is to recover these masked tokens given the surrounding context. Encoder-only models such as \codebert~\cite{feng2020codebert} and \graphcodebert~\cite{guo2021graphcodebert} are designed to provide a representation of the input code to be used for downstream tasks such as code classification~\cite{yang2020xlnet}. Decoder-only models, on the other hand, aim to autoregressively generate tokens based on all previously generated tokens. \codegen~\cite{Nijkamp2022CG,nijkamp2023codegen2}, \codex~\cite{chen2021codex} and \polycoder~\cite{xu2022systematic} are examples of decoder-only \llm{s} where they can be used for code autocompletion tasks.
Different from encoder- and decoder-only \llm{s}, encoder-decoder models (e.g., \codetf~\cite{wang2021codet5,wang2023codet5} and \plbart~\cite{ahmad2021PLBART}) combine both encoder and decoder together and jointly train both components together. A commonly used pre-training objective for encoder-decoder models is \mspfull (\msp) where random spans (multiple consecutive tokens) are replaced with single masked tokens and the models learn to \emph{fill in} the masked span with the correct sequence of tokens. 
Furthermore, decoder-only models like \incoder~\cite{fried2023incoder} can also do infilling through the causal language modeling~\cite{aghajanyan2022cm3} objective. Instead of using the decoder to predict the next token in the original training data, similar to \msp, \incoder also replaces random spans with masked span tokens.
During training, \incoder learns to autoregressively recover the original spans. With this training strategy, \incoder can perform infilling with bidirectional context similar to encoder-decoder models, enabling cloze-style repair.

\subsection{Code Completion}
Code completion is one of the most frequently used features in \idefull{s} (\ide{s}).
It substantially alleviates the complexity of software development by
interactively suggesting program constructs after the user's \caret{} position
while programmers are typing,
including identifier names and library APIs.
Code completion is now an indispensable infrastructure of the most widely-used programming languages
and can be easily integrated into most modern text editors
thanks to the
presence of the Language Server Protocol~\cite{lsp}, which standardizes the communication 
between tools and language services.
Traditionally, a \emph{semantics-based} \emph{\cefull} is implemented on top of
a series of complex \emph{incremental} syntactic and semantic analyses
of the target programming language,
since it needs to understand partially written programs and provide real-time feedback.
The \cefull has full access to a project repository and its dependencies
and can produce suggestions according to its semantic understanding.
Recent advances in \llm{s} demonstrate the capability of generating
long and complicated completions.
However, they may produce unreasonable programs due to the limitation in the code 
context size and the loss of program analysis by simply treating programs as token sequences.
In this paper, we use the term \cefull to refer to the
\emph{semantics-based} one.
We formally define the expected properties of a \cefull in our framework in \Cref{def:strictCe}.

\subsection{\aprfull}

\aprfull (\apr) aims to generate patches given the buggy code location and the bug-exposing tests. Traditionally, \apr approach can be categorized as constraint-based~\cite{demacro2014nopol, mechtaev2016angelix, le2017s3, long2015spr}, heuristic-based~\cite{legoues2012genprog, wen2018capgen, le2016hdrepair} and template-based~\cite{hua2018sketchfix, martinez2016astor, koyuncu2020fixminder, liu2019tbar, liu2019avatar, ghanbari2019prapr}. Among these classic techniques, template-based tools have been shown to achieve the highest number of bug fixes by using handcrafted repair templates to target specific bug patterns~\cite{ghanbari2019prapr}. However, these handcrafted patterns cannot cover all types of bugs that exist and as such, template-based tools cannot fix bugs outside of their pre-determined templates. 

To address the issue faced by template-based \apr tools, researchers resort to \nmtfull (\nmt)~\cite{sutskever2014mt} to develop \nmt{}-based \apr tools~\cite{ye2022rewardrepair, zhu2021recoder, li2020dlfix, lutellier2020coconut, jiang2021cure, chen2018sequencer}. \nmt{}-based \apr tools train an \nmt model to translate the input buggy code into the correct code through bug-fixing datasets containing pairs of buggy and fixed code. However, these bug-fixing datasets may contain only a small number/types of bug fixes, especially compared to a large amount of available open-source code snippets, due to the difficulty in obtaining bug-fixing commits~\cite{xia2022alpharepair}. Additionally, the datasets can fail to filter out unrelated commits~\cite{jiang2021extract} such as refactoring, which adds noise to the training datasets. Due to this reliance on training using bug-fixing datasets, these \nmt{}-based tools also cannot generalize to bug types not seen during training. 

Recently, researchers begin to directly apply \llm{s} for \apr~\cite{xia2022repairstudy}. \alpharepair~\cite{xia2022alpharepair} is the first to directly use \llm{s} for \textit{cloze-style} (or infilling-style) \apr: it masks out the buggy code snippet and then uses \codebert~\cite{feng2020codebert} to directly fill in the correct code given the surrounding context. While \alpharepair demonstrates the potential to use encoder-only models for cloze-style \apr, other studies~\cite{xia2022repairstudy, prenner2021codexws, kolak2022patch} have looked into applying all three types of \llm architecture.
\fitrepair{}~\cite{xia2023revisiting} further improves \alpharepair via domain-specific fine-tuning and prompting strategies leveraging the plastic surgery hypothesis~\cite{Barr14fse}. Even more recently, researchers have applied dialogue-based models for APR~\cite{xia2023conversational, xia2023conversation,sobania2023analysis,cao2023study}. For example,
\chatrepair{}~\cite{xia2023conversation} proposes a fully automated conversational APR approach by learning from prior patching attempts, including both patch code and test failure information. %

Compared to traditional and NMT-based \apr techniques, \llm-based techniques are able to achieve new state-of-the-art bug-fixing results~\cite{xia2022alpharepair, xia2022repairstudy}. While the performance is impressive, one particular limitation of these techniques is the lack of guidance in patch generation. Prior work mainly treats the \llm as a black box and only queries the model via beam search~\cite{xia2022alpharepair} or sampling~\cite{xia2022repairstudy, prenner2021codexws, kolak2022patch}. This means \llm{s}, while powerful, may still generate invalid patches given the current code context.

In this work, we address these limitations by using a \emph{semantics-based} \cefull to guide and prune the \llm{} search space. Our approach is orthogonal to recent \llm-based \apr techniques and can be easily combined with them. In fact, NMT-based \apr techniques have also attempted to tackle this problem.
\cure~\cite{jiang2021cure} first statically obtains the valid identifiers and forces the \nmt model to only select from valid identifiers during generation.
\recoder~\cite{zhu2021recoder} builds an edit-based \nmt model to enforce syntax correctness and introduce placeholder tokens and then as a post-processing step, \recoder will replace placeholder tokens with statically determined valid identifiers.
\rewardrepair~\cite{ye2022rewardrepair} on the other hand, attempts to increase the number of compilable patches by penalizing uncompilable patches during training. 
Compared to these prior techniques, \tool is more general and effective. \tool{} does not require any domain-specific training and
leverages the incremental analysis of off-the-shelf \cefull{s} to enforce guaranteed constraints to guide \llm{s} on the fly.

\section{Preliminaries}

\label{sec:formalization}
In this section, we first define concepts about programming languages used throughout the paper (\Cref{subsec:lang}).
Then we discuss
the \emph{formal abstractions} of the two key components used in our \tool{} framework:
\cefull{} (\Cref{subsec:absCE}) and \plmfull{} (\Cref{subsec:absLLM}).
These two abstractions are crucial in that each of them describes a collection of fitting \emph{implementations},
which forms the reason why \tool{} is a generalizable framework.

\subsection{Languages with Static Checking}
\label{subsec:lang}
We now introduce the concept of programming languages equipped with static checking
and define the feasibility of a partial program
before the formulation of the \cefull (\Cref{def:ce}).

\begin{definition}[Programming Language with Static Checking]
A programming language with static checking is defined as
a pair of its character set \fmCeAlphabet{} and its static specification 
$\fmStaticChecker{}\subseteq \fmProgType{}$ as a unary relation on \fmProgType{}.
\begin{equation}
\fmProgWithCheckEq{},
\end{equation}
Given a $\fmProg\in\fmProgType{}$, the notation $\fmStaticChecker(\fmProg)$
(or $\fmProg\in\fmStaticChecker$)
states that \fmProg{} is a statically valid program in this language.
For statically-typed programming languages like Java, the compilation check is a kind of static checking.
\end{definition}

\begin{definition}[Static Feasibility of A Partial Program]
For a partially written program $\fmProg{}\in\fmProgType$,
we say it is feasible at the \caret{} position \fmCeCaret{}
with respect to the static specification \fmStaticChecker{},
written as $\fmFeasible{(\fmProg{}, \fmCeCaret{})}$,
if and only if there exists a possible continuation after \fmCeCaret{} with which completing \fmProg{} results in a statically valid program.
The definition can be formally written as
\begin{equation}
\fmFeasible{(\fmProg{}, \fmCeCaret{})} \triangleq \exists \fmCont{} \in \fmProgType{}, \fmStaticChecker(\fmProg\left[\fmCeCaret\leftarrow\fmCont\right]),
\end{equation}
where we use the notation $\fmProg\left[\fmCeCaret\leftarrow\fmCont\right]$ as
the action of completing \fmProg{} at \fmCeCaret{} with \fmCont, i.e.
\begin{equation}
\fmProg\left[\fmCeCaret\leftarrow\fmCont\right] \triangleq \fmProg_{0..\fmCeCaret} \cdot \fmCont \cdot \fmProg_{\fmCeCaret{}..|\fmProg|}.
\end{equation}
In \Cref{algo:overall}, we extend this notation to accept a $\fmRange\fmOfType{}\mathbb N \times \mathbb N$,
so that $\fmProg\left[\fmRange\leftarrow\fmGen\right]$ specifies the action of replacing \fmProg{}'s contents within \fmRange{} with \fmGen{}.
\end{definition}

\subsection{Abstraction of Completion Engines}
\label{subsec:absCE}
A \cefull, showed in \Cref{fig:ce}, provides suggested continuations to
a partially written program given the \caret{} position.

\begin{definition}[\cefull]
\label{def:ce}
Formally speaking, a \cefull \fmCe{} is a pair
\begin{equation}
    \fmCe = (\fmCeAlphabet, \fmCeComplete),
    \label{eq:ce}
\end{equation}
where \fmCeAlphabet{} is the character set of the target language, and
\begin{equation}
    \fmCeComplete\fmOfType{}(\fmProgType{}, \mathbb N) \to \fmCeCompRet \cup \{\fmCeNone\}
    \label{eq:ce:getContext}
\end{equation}
is a function to obtain the completions given a program at some \caret position, with \fmCeNone{} indicating the engine cannot determine the suggestions from the code context
(e.g., when completing a variable declaration).
Note that we make a distinction between \fmCeNone{} and empty completions $\varnothing$
because in this paper we are interested in a specific group of \emph{strict} \cefull{s} that
helps determine the feasibility of a partial program.

\begin{figure}[htbp]
\centering
\includegraphics[width=0.8\linewidth,page=3]{Drawing.pdf}
\caption{Abstraction of a \cefull.}
\label{fig:ce}
\end{figure}

\end{definition}

\begin{definition}[Strict \cefull{}]
\label{def:strictCe}
Assume that a \cefull \fmCe{} can obtain a
set of completions given a program \fmProg{} feasible at \fmCeCaret{} (i.e., $\fmFeasible{(\fmProg, \fmCeCaret)}$):
\begin{equation}
\begin{gathered}
\fmCompletions = \fmCeComplete(\fmProg, \fmCeCaret)\\
\text{where }\fmCompletions \ne \fmCeNone.
\end{gathered}
\end{equation}
Then, \fmCe{} is said to be strict %
if and only if, under this condition,
continuing \fmProg{} with any code that does not match with this set of completions yields an infeasible program at the new \caret{} position:
\begin{equation}
\begin{gathered}
    \forall \fmCompVar\not\in\fmPrefixSet(\fmCompletions), \fmInfeasible{(\fmProg', \fmCeCaret')},\\
    \text{where }\fmProg' = \fmProg\left[\fmCeCaret\leftarrow\fmCompVar\right]
    \text{ and }\fmCeCaret' = \fmCeCaret + |\fmCompVar|,\\
    \fmPrefixSet({\cdot}) = \{c \mid s \in {\cdot} \text{ and } c\text{ is a prefix of }s\text{ or vice versa}\}.
    \label{eq:strictCe}
\end{gathered}
\end{equation}
This definition essentially means that a strict \cefull should not
give incorrect suggestions. It should return \fmCeNone{} whenever unsure.
A trivial strict \cefull can be the one that always returns \fmCeNone{}.
\end{definition}
\subsection{Abstraction of \llm{s}}
\label{subsec:absLLM}

In this section, we give a formal abstraction of an encoder-decoder based \plm{} as showed in \Cref{fig:llm},
which in practice is more complex but conforms to the abstraction.
The abstraction subsumes decoder-only models and can also describe encoder-only models that use the encoder outputs directly as token probabilities for generation.

\begin{figure}[htbp]
\centering
\includegraphics[width=0.75\linewidth,page=2]{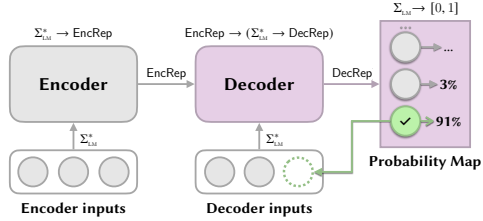}
\caption{Abstraction of encoder-decoder based \plm{}.}
\label{fig:llm}
\end{figure}

\begin{definition}[\plmfull{}]
\label{def:llm}
Formally, We define an encoder-decoder based \plm{} \fmPlm{} as a 3-tuple 
\begin{equation}
    \fmPlm = (\fmPlmTok, \fmEnc, \fmDec),
    \label{eq:llm}
\end{equation}
where \fmPlmTok{} is a \emph{vocabulary} consisting of the set of tokens defined by the model.
The encoder \fmEnc{} is a function that maps from an input sequence to its encoded representation in \fmEncRep{}:
\begin{equation}
\fmEnc \fmOfType{}\fmSeq{\fmPlmTok} \to \fmEncRep.
\end{equation}
The decoder \fmDec{}, defined below, then
\emph{memorizes} the encoded representation in \fmEncRep{}, takes as input a sequence of tokens, and produces as output its decoded representation in \fmDecRep{}:
\begin{equation}
\fmDec\fmOfType{}\fmEncRep \to \left(\fmSeq{\fmPlmTok} \to \fmDecRep\right).
\end{equation}
In this definition, the decoder memorizing the encoded representation is modeled as a \emph{higher-order function} that returns a detailed decoding function given the encoded representation.
The decoded representation in \fmDecRep{} essentially assigns a probability to each token in the vocabulary to state its likelihood of being the next token in the sequence. Therefore, we can \textit{define} \fmDecRep{} as
\begin{equation}
\fmDecRepDef.
\label{eq:decRepDef}
\end{equation}
\end{definition}

\section{Approach}

Following most recent deep learning based \apr{} tools~\cite{xia2022alpharepair, zhu2021recoder, ye2022rewardrepair},
\tool{} focuses on fixing single-hunk bugs, where the patch is obtained by changing a continuous section of code under perfect fault localization. \tool{} can be extended for multi-hunk bugs by replacing all hunk locations at the same time with separate infilling tokens and using \llm to generate the replacement hunks. 
Benefiting from the era of \llm{s},
as shown in \Cref{fig:cloze}, in this paper, we treat the repair problem as a \emph{cloze} task~\cite{xia2022alpharepair},
where a patch is formed by first replacing the buggy hunk with a masked span token (\CodeIn{<SPAN>}) and then using the \llm to direct synthesize the fixed hunk from the surrounding code context to replace the span token.
\begin{figure}[!htbp]
\centering
\includegraphics[width=\linewidth]{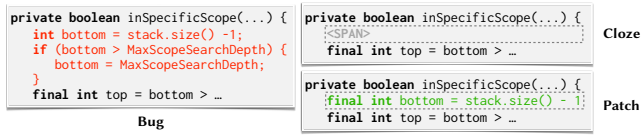}
\caption{Cloze-style program repair.}
\label{fig:cloze}
\end{figure}

\subsection{Overview}
\begin{figure*}[!t]
\centering
\includegraphics[page=1,width=0.85\linewidth]{Drawing.pdf}
\caption{Overview of \tool{}.}
\label{fig:overview}
\end{figure*}
\Cref{fig:overview} shows an overview of how \tool{} synthesizes a program
that acts as the repaired hunk of the original buggy program.
The generation loop consists of a loop that keeps updating the generation with tokens
newly generated from the synergy between the language model and \cefull.
The loop starts by applying the current generation as the input to the language model (\circled{1}),
which returns a search space of a mapping from a suggested next token to its probability.
\tool{} then enters a token selection phase that repeatedly
samples a token from the search space, checking its feasibility, and pruning the search space
until a token is accepted.
Every time a token is sampled, \tool{} first checks if it hits the memorization (\circled{2}),
which stores the tokens that are known to be feasible or infeasible.
The memorization of infeasible tokens includes the use of a prefix tree data structure (Trie)
discussed in \Cref{subsec:memorize}.
When the token hits the memorization and is infeasible, the search space is pruned
by setting this token's probability to zero (\circled{3}),
and the next sampling will run on the updated search space.
In this way, the same token is not sampled again during the token selection phase.
If the token misses the memorization, the search space is pruned under the guidance of
the \cefull (\circled{4}), which we elaborate in \Cref{subsec:prune}.
Provided that the sampled token is rejected by the \cefull,
\tool{} zeroes out its probability. Otherwise, it is accepted and this
token selection process terminates.
The memorization gets updated in both cases (\circled{5}).
After a token is accepted (\circled{6}), we further leverage the \cefull,
trying to actively complete the token (\circled{7}).
The active completion, discussed in \Cref{subsec:active},
may either produce more tokens or add nothing to the accepted token.
Finally, \tool{} appends all the newly generated tokens to the current generation and begins
a new loop until a complete patch is generated. 
The loop stops when the model generates the special token \fmEndToken{}.

\Cref{algo:overall} details this process
and shows how a complete patch program is generated
using what is established in \Cref{sec:formalization}.
It additionally describes how \tool{} performs the pre-processing (\Crefrange{algo:overall:buildInput}{algo:overall:initGen})
and formalizes completion-guided pruning procedure illustrated in \Cref{fig:overview}
using two functions \fmPruneDecode{} and \fmActiveComplete{} (\Crefrange{algo:overall:loopBegin}{algo:overall:updateGen}).
In all our algorithms, we use a "dot-notation" to specify an entity of a tuple (e.g., $\fmPlm{}.\fmEnc{}$),
but use an abbreviation form when the context is clear
(e.g., $\fmPlmTok$ and $\fmCeAlphabet$ for $\fmPlm{}.\fmPlmTok{}$ and $\fmCe.\fmCeAlphabet$).
We also optionally apply \emph{type annotations} for clarification. Note that we simplify the definition of the \cefull by restricting it to be called with one program. In practice, a \cefull is always initialized with the entire project and can provide suggestions based on global information.

\begin{algorithm}
\caption{Main Repair Loop of \tool}\label{algo:overall}
\algoSize

\begin{algorithmic}[1]
\fmAlgoInputs{\plmfull{} \fmPlm{}, \cefull{} \fmCe{}, Buggy program \fmProg, and Range of buggy hunk \fmRange.}
\fmAlgoOutput{Patch for the buggy program.}
\algoSpacing

\Func{\fmRepair}{\fmPlm{}, \fmCe{}, $\fmProg{}\fmOfType{}\fmProgType$, $\fmRange\fmOfType{}\mathbb N \times \mathbb N$}{\fmProgType{}}
\LComment{Initializations based on \Cref{def:llm}}
\State $\fmPlmInputs\fmOfType{}\fmSeq{\fmPlmTok} \gets \fmBuildInputs(\fmProg{}, \fmRange{})$ \label{algo:overall:buildInput}
\State $\fmPlmEncRep\fmOfType{}\fmEncRep \gets \fmPlm.\fmEnc(\fmPlmInputs)$ \label{algo:overall:getEncodeRep}
\State $\fmDecodeDetail\fmOfType{}\fmSeq{\fmPlmTok} \to \fmDecRep{} \gets \fmPlm.\fmDec(\fmPlmEncRep)$ \label{algo:overall:getDecode}
\State $\fmGen\fmOfType{}\fmSeq{\fmPlmTok} \gets \varepsilon$ \label{algo:overall:initGen}
\While{\fmKwTrue} \label{algo:overall:loopBegin}
    \LComment{Form \fmPatch{} by replacing buggy hunk with \fmGen}
    \State $\fmPatch\gets \fmProg\left[\fmRange \leftarrow \fmConcat(\fmGen) \right]$ \label{algo:overall:getPatch}
    \LComment{Move \fmCeCaret{} after the current generation}
    \State $\fmCeCaret\gets \fmRange.\fmRangeStart + |\fmConcat(\fmGen)|$ \label{algo:overall:getCaret}
    \State $\fmTokens\fmOfType{}\fmDecRepRhs{} \gets \fmDecodeDetail(\fmGen)$ \label{algo:overall:getTokens}
    \State $\fmNextToken\fmOfType\fmPlmTok{} \gets \fmPruneDecode(\fmCe, \fmPatch, \fmCeCaret, \fmTokens)$ \label{algo:overall:prunedSearch}
    \If{$\fmNextToken = \fmEndToken{}$} \label{algo:overall:CheckEnd}
        \State \Return{\fmPatch} \label{algo:overall:return}
    \EndIf
    \State $\fmCompletionTokens{}\fmOfType\fmSeq{\fmPlmTok} \!\gets\! \fmActiveComplete(\fmCe, \fmPatch, \fmCeCaret)$ \label{algo:overall:activeComplete}
    \State $\fmGen \gets \fmGen \cdot \fmNextToken \cdot \fmCompletionTokens$ \label{algo:overall:updateGen}
\EndWhile
\EndFunc
\end{algorithmic}
\end{algorithm}

\subsection{\coreApproachDef}
\label{subsec:prune}
In this section, we explain the core idea of how \tool{} utilizes a \cefull{} to prune the search space of an \llm{}.

\Cref{algo:prune} explains in depth how a \cefull helps prune the model's search space.
The function \fmPruneDecode{} takes as inputs a \cefull \fmCe{}, the current program \fmProg{},
the current \caret{} position \fmCeCaret{}, and the probability map \fmTokens{} given by the model,
and produces a token \fmNextToken{} as the continuation of the program \fmProg{} at position \fmCeCaret{}.
The function consists of a \textbf{while}-loop (\Crefrange{algo:prune:loopStart}{algo:prune:return}) where
\tool{} first samples a possible next token according to the probabilities (\Cref{algo:prune:sample}),
updates the current program accordingly (\Cref{algo:prune:updateProg}),
and moves the \caret{} after \fmNextToken{}.
\tool{} then invokes the \cefull using the function \fmCeComplete{} defined in \Cref{eq:ce:getContext}, given the program $\fmCeProg'$ and the \caret{} position $\fmCeCaret'$.
If the result is not \fmCeNone{} but there is no completion (\Cref{algo:prune:testContext}),
it means that no possible continuation can be formed after \fmNextToken{},
so the token \fmNextToken{} is considered infeasible,
thus \emph{pruned} (\Cref{algo:prune:prune}) in this round of search,
and the loop will continue (\Cref{algo:prune:continue}).
Otherwise, we consider the token feasible and return \fmNextToken{} (\Cref{algo:prune:return}).

The pruning at \Cref{algo:prune:prune} is done by setting the probability of the entry \fmNextToken{} of the probability map \fmTokens{} to zero.
The notation used at this line is defined subsequently.
Assume that
\begin{equation}
\fmArbFunc\fmOfType\fmArbFuncType = \{\fmArbIn_0 \mapsto \fmArbOut_0, \fmArbIn_1 \mapsto \fmArbOut_1, \dots\}
\end{equation}
is an arbitrary function, and 
\begin{equation}
\fmArbPartialFunc{}\fmOfType\fmArbPartialFuncType = \{\fmArbIn'_0 \mapsto \fmArbOut'_0, \fmArbIn'_1 \mapsto \fmArbOut'_1, \dots\}
\end{equation}
is a \emph{partial function} of the same type,
meaning that only a subset of inputs in the domain \fmArbDom{} is associated with an output
in the range \fmArbRange{}.
We define the action of changing the output values of the inputs in \fmArbFunc{} using the assignments given by \fmArbPartialFunc{} as
\begin{equation}
\begin{aligned}
    \fmArbFunc\left[\fmArbPartialFunc\right] &=
    (\fmArbFunc - \fmRemoved) \cup
    \fmArbPartialFunc{} \\
    \text{where }\fmRemoved &= \{\fmArbIn' \mapsto \fmArbFunc(\fmArbIn') \mid
    \fmArbIn' \mapsto \fmArbOut' \in \fmArbPartialFunc\}.
\end{aligned}
\end{equation}

\begin{algorithm}
\caption{\coreApproachDef}
\label{algo:prune}

\algoSize
\begin{algorithmic}[1]
\fmAlgoInputs{\cefull{} \fmCe{}, Current Program \fmProg{}, Caret Position \fmCeCaret{}, and Token Probability Map \fmTokens{}.}
\fmAlgoOutput{Next token $\fmNextToken$ to generate.}
\algoSpacing

\Func{\fmPruneDecode}{\fmCe{}, \fmProg{}, \fmCeCaret{}, $\fmTokens\fmOfType\fmDecRepRhs{}$}{\fmPlmTok}
\While{\fmKwTrue} \label{algo:prune:loopStart}
    \State $\fmNextToken\fmOfType{}\fmPlmTok{} \gets \fmSample(\fmTokens)$ \label{algo:prune:sample}
    \State $\fmProg' \gets \fmProg\left[\fmCeCaret\leftarrow \fmConcat(\fmNextToken)\right]$ \label{algo:prune:updateProg}
    \State $\fmCeCaret' \gets \fmCeCaret + |\fmConcat(\fmNextToken)|$  \label{algo:prune:moveCaret}
    \LComment{$\fmCompletions\fmOfType\fmPow{\fmProgType} \cup \{\fmCeNone\}$}
    \State $\fmCompletions \gets \fmCe.\fmCeComplete(\fmProg', \fmCeCaret')$ \label{algo:prune:getContext}
    \If{$\fmCompletions \ne \fmCeNone$ \fmKwAnd{} $|\fmCompletions| = 0$} \label{algo:prune:testContext}
        \State $\fmTokens \gets \fmTokens\left[\{\fmNextToken \mapsto 0\right\}]$ \label{algo:prune:prune}
        \State \fmKwContinue{} \label{algo:prune:continue}
    \EndIf
    \State \Return{\fmNextToken} \label{algo:prune:return}
\EndWhile
\EndFunc
\end{algorithmic}
\end{algorithm}

\subsection{Memorization for Faster Search}
\label{subsec:memorize}
\Cref{algo:prune} (\fmPruneDecode{}) involves a loop of trials and pruning actions, which
slows down the repair task in some situation.
To speedup its search procedure, we apply several memorization techniques to
reduce the frequency of invoking the \cefull for analysis.

\paragraph{Memorizing rejected tokens}
To repair a bug in practice requires generating plenty of samples,
meaning that the same program $\fmProg'$ and $\fmCeCaret'$ (\Crefrange{algo:prune:updateProg}{algo:prune:moveCaret})
may occur repeatedly in \Cref{algo:prune} (\fmPruneDecode{}).
Therefore, we can memorize all the tokens pruned at \Cref{algo:prune:prune}
by storing them in a variable
\begin{equation}
\fmRejected\fmOfType (\fmProgType, \mathbb N) \to \fmPow{\fmPlmTok},
\end{equation}
which maps from a program \fmProg{} and a \caret{} position \fmCeCaret{} to a set of rejected tokens.
Then we zero the probabilities of the rejected tokens in advance,
written as
\begin{equation}
\fmTokens{}\gets\fmTokens\left[\{\fmToken \mapsto 0 \mid \fmToken \in \fmRejected(\fmProg, \fmCeCaret{})\}\right],
\end{equation}
before the \textbf{while}-loop (\Cref{algo:prune:loopStart}) starts.

\paragraph{Memorizing accepted tokens}
Besides rejected tokens, we can also memorize tokens that are accepted before in a variable
\begin{equation}
\fmAccepted\fmOfType (\fmProgType, \mathbb N)\to \fmPow{\fmPlmTok}
\end{equation}
to avoid the overhead incurred from querying the \cefull
at \Crefrange{algo:prune:getContext}{algo:prune:testContext}.

\paragraph{Building a Prefix Tree of Rejected Tokens}
It is common that many tokens in the vocabulary of the language model are prefixes of another.
And it is obvious that if a token is rejected,
meaning that no possible continuation can be formed after the token to obtain a statically valid program,
then any token sharing such prefix should be rejected.
For this reason, we build and keep updating a prefix tree, or Trie,
of all the rejected tokens given \fmProg{} and \fmCeCaret{},
and checks if any of the tokens in the Trie is a prefix of \fmNextToken{}
right after \Cref{algo:prune:sample} in \Cref{algo:prune}.
If it is the case, \tool{} directly skips to the next iteration, avoiding further analysis.

\subsection{Active Completion}
\label{subsec:active}
Not only is a \cefull able to determine the feasibility of a possible next token
suggested by the model, as shown in \Cref{subsec:prune},
but it can also \emph{proactively} suggest a potential continuation of the current program
without querying the model,
just like how developers benefit from autocompletion.

\Cref{algo:active} describes active completion in detail.
The function \fmActiveComplete{} takes three inputs:
the \cefull \fmCe{}, the current program \fmProg{}, and the current \caret{} position \fmCeCaret{},
and outputs a sequence of tokens \fmCompletionTokens{} as the continuation of \fmProg{} at \fmCeCaret{}.
Initially, \tool{} gets the completion result according to \Cref{eq:ce:getContext}, given \fmProg{} and \fmCeCaret{} (\Cref{algo:active:getContext}),
and checks if it is \fmCeNone{} (\Cref{algo:active:testContext}).
If it is the case ($\fmCompletions{} = \fmCeNone$), the result is set to an empty string,
meaning no extra completions are produced (\Cref{algo:active:default}).
Otherwise, \tool{}
calculates the \emph{common prefix} of all the completions (\Cref{algo:active:commonPrefix}).
Note that the type of the resultant variable \fmContinuation{} is a sequence of characters in the Programming Language alphabet,
different from the language model's $\fmPlmTok{}$,
so \tool{} further aligns the completion to fit the model's vocabulary (\Cref{algo:active:tokenize}). Finally, the result is returned at \Cref{algo:active:return}.
\begin{algorithm}
\caption{Active Completion}
\label{algo:active}

\algoSize
\begin{algorithmic}[1]
\fmAlgoInputs{\cefull{} \fmCe{}, Program \fmProg, and Caret Position \fmCeCaret.}
\fmAlgoOutput{The actively completed tokens \fmCompletionTokens{}.}
\algoSpacing

\Func{\fmActiveComplete{}}{\fmCe{}, \fmProg{}, \fmCeCaret{}}{\fmSeq{\fmPlmTok}}
\State $\fmCompletions\fmOfType\fmPow{\fmProgType} \cup \{\fmCeNone\}\gets \fmCe.\fmCeComplete(\fmProg, \fmCeCaret)$ \label{algo:active:getContext}
\If{$\fmCompletions = \fmCeNone$} \label{algo:active:testContext}
    \State $\fmCompletionTokens{} \gets \varepsilon$ \label{algo:active:default}
\Else
    \State $\fmContinuation\fmOfType\fmSeq{\fmCeAlphabet} \gets \fmCommonPrefix(\fmCompletions)$ \label{algo:active:commonPrefix}
    \State $\fmCompletionTokens\fmOfType{}\fmSeq{\fmPlmTok} \gets \fmTokenDecode(\fmPlmTok{}, \fmContinuation{})$ \label{algo:active:tokenize}
\EndIf
\State \Return{\fmCompletionTokens{}} \label{algo:active:return}
\EndFunc
\end{algorithmic}
\end{algorithm}

\subsection{Soundness of \tool{}}
In this section, we show the theoretical guarantee of each algorithm discussed above
under the condition that the \cefull is \emph{strict} (\Cref{def:strictCe}).

\begin{lemma}[Soundness of Pruning]
\label{thm:prune}
The tokens pruned away in \Cref{algo:prune} (\fmPruneDecode{}) result in infeasbile programs.
\end{lemma}

\begin{proof}
From \Cref{eq:strictCe} in \Cref{def:strictCe}, we can deduce that a program is infeasible at some \caret{} position
if the \cefull does not return \fmCeNone{} but the set of completions is empty, i.e.,
\begin{equation}
\begin{gathered}
|\fmCompletions| = 0 \to \fmInfeasible{(\fmProg{}, \fmCeCaret{})}\\
\text{if }\fmCompletions\ne\fmCeNone
\end{gathered}
\end{equation}
The pruning at \Cref{algo:prune} happens at \Crefrange{algo:prune:testContext}{algo:prune:prune},
which is exactly what is described above.
As a result, we can conclude that the program with \fmNextToken{} appended is infeasible, and
hence it is safe for \tool{} to abandon the token.
\end{proof}

\begin{lemma}[Soundness of Memorization]
\label{thm:mem}
The memorization discussed in \Cref{subsec:memorize} does not affect
\fmPruneDecode{}'s behavior.
\end{lemma}
\begin{proof}
The theorem holds because all the memorization techniques mentioned in \Cref{subsec:memorize} do not
change the semantics of \fmPruneDecode{} but only speed up the process.
\end{proof}

\begin{lemma}[Soundness of Active Completion]
\label{thm:active}
If a program is feasible at some \caret{} position, the new program
produced by \Cref{algo:active} (\fmActiveComplete{})
is feasible at its new caret position.
\end{lemma}
\begin{proof}
Based on \Cref{eq:strictCe} from \Cref{def:strictCe}, any continuations not matching the set of completions would bring about an infeasible program. In the case where these completions have a shared common prefix, any continuations not starting with this common prefix would be invalid. Therefore, completing the  original program with the common prefix (\Cref{algo:active:commonPrefix} in \Cref{algo:active}) is the only way to yield a new feasible program.
\end{proof}

On the basis of \Crefrange{thm:prune}{thm:active}, we can easily prove that \tool{}'s overall
algorithm is sound.
\begin{theorem}[Overall Soundness]
\Cref{algo:overall} (\fmRepair) does not miss any feasible programs in the language model's
search space.
\end{theorem}

\paragraph{When will \tool{} fail?}
Although the theorems are about the soundness of \tool{}, i.e., it \emph{prunes the search space correctly},
it does not provides any guarantee that \tool{} \emph{produces a valid patch} every time.
Therefore, \tool{}'s expected behavior is to be able to obtain valid patches more efficiently,
rather than being entirely error-free during the generation.

\section{Experimental Setup}

In this paper, we study the following research questions to evaluate \tool{}.
\begin{itemize}
    \item \textbf{RQ1:}  How does \tool{}'s bug fixing capability compare with state-of-the-art \apr techniques (\Cref{subsec:resultComparison})?
    \item \textbf{RQ2:}  How effective is \tool{} in improving the compilation rate of patch generation (\Cref{subsec:patchEfficiency})?
    \item \textbf{RQ3:}  Are all components of \tool{} making positive contributions to its effectiveness (\Cref{subsec:ablation})?
    \item \textbf{RQ4:} Can \tool generalize to different subjects of bugs and models (\Cref{subsec:generalizability})?
\end{itemize}

We first compare the repair performance of \tool{}, instantiated with \codetf, against state-of-the-art \apr tools across both traditional, \nmt{}-based, and \llm-based tools on the \dfj datasets in RQ1. 
In RQ2, we then closely evaluate the improvement in compilation rate --- percentage of compilable patches generated to demonstrate that \tool is not only effective in bug repair but can generate a higher number of compilable patches compared with existing tools. Furthermore, we perform a detailed ablation study in RQ3 to evaluate the contribution of different components of \tool.
Finally, in RQ4, we extend our evaluation of \tool beyond its use with \codetf in the previous RQs. We go a step further by implementing \tool with \incoder and assessing the performance of \tool using both \codetf and \incoder on single-hunk bugs from both \dfj 1.2 and 2.0 to demonstrate the generalizability of \tool across different \llm{s} and bug subjects.

\subsection{Implementation}
\label{sec:impl}
We use the Python implementation of the \codetf-large and the \incoder-6.7B models obtained on Hugging Face~\cite{HuggingFaceWebPage}. 
We build our generation pipeline in Python with 5K lines of code and
implement a modified version of the Eclipse JDT Language Server~\cite{eclipseLS,modifiedJDT} in Java with 1.5K additional lines of code, which serves as the \emph{strict} \cefull of our framework.
Our default generation uses top-p nucleus sampling~\cite{holtzman2019nucleus} with $p=1.0$, $\mathit{temperature} = 1.0$, $max\fmHyphen tokens = 50$
and samples 5000 times per bug for fair comparisons against prior works (\Cref{subsec:resultComparison} and \Cref{subsec:patchEfficiency}). Due to the high cost of \apr, we sample 500 times per bug for the ablation study (\Cref{subsec:ablation}) and the generalizability evaluation (\Cref{subsec:generalizability}). 
Following prior work~\cite{lutellier2020coconut, zhu2021recoder, li2020dlfix, xia2022alpharepair}, we use a timeout of 5 hours to generate and validate all patches per bug.
We generate and validate patches on a 32-Core with Ryzen Threadripper PRO 3975WX CPU, 256 GB RAM and NVIDIA RTX A6000 GPU, running Ubuntu 20.04.4 LTS with Java version OpenJDK 1.8.0\_181. 

\subsection{Subject Programs}

We use the popular repair benchmark of \dfj{}
for our evaluation. \dfj is a manually curated Java dataset with pairs of buggy and patched versions of the source project along with developer test suites for validation. Following prior work and \apr literature convention, we separate \dfj into \dfj 1.2, containing 391 bugs (removing 4 depreciated bugs) from 6 Java source projects, and \dfj 2.0, containing 438 new bugs from 9 additional projects. For \dfj 1.2, we focus on only the single-hunk bugs as \tool is designed for single-hunk repair. Note this is also the evaluation setting used in the prior baseline~\cite{ye2022rewardrepair}. Furthermore, we remove the bugs that are incompatible with our \cefull due to engineering issues. In total, we consider 138 single-hunk bugs from \dfj 1.2 and 135 single-hunk bugs from \dfj 2.0.
For our main evaluation in RQ1, following the same setup as prior \llm for \apr work~\cite{xia2022alpharepair, xia2022repairstudy}, we report the results on all 135 single-hunk bugs from \dfj 1.2 and 76 single-line bugs (a subset of single-hunk bugs) from \dfj 2.0. Meanwhile, in our generalizability study (RQ4), we further evaluate \tool on the full set of single-hunk bugs from both \dfj 1.2 and 2.0 for both \codetf and \incoder.

\subsection{Compared Techniques}

We compare \tool against state-of-the-art baselines from traditional, \nmt{}-based, and \llm for \apr tools. We evaluate against \alpharepair~\cite{xia2022alpharepair} as it is the top performing \llm for \apr approach. For \nmt{}-based approaches, we choose 6 recent tools: \rewardrepair~\cite{ye2022rewardrepair}, \recoder~\cite{zhu2021recoder}, \cure~\cite{jiang2021cure}, \coconut~\cite{lutellier2020coconut}, \dlfix~\cite{li2020dlfix} and \sequencer~\cite{chen2018sequencer} based on the \nmt architecture. Additionally, we compare against 12 traditional \apr tools: \prapr~\cite{ghanbari2019prapr}, \tbar~\cite{liu2019tbar}, \avatar~\cite{liu2019avatar}, \simfix~\cite{jiang2018simfix}, \fixminer~\cite{koyuncu2020fixminder}, \capgen~\cite{wen2018capgen}, \jaid~\cite{chen2017jaid}, \sketchfix~\cite{hua2018sketchfix}, \nopol~\cite{demacro2014nopol}, \jgenprog~\cite{martinez2015automatic}, \jmutrepair~\cite{martinez2016astor} and \jkali~\cite{martinez2016astor}. Altogether, we include 19 \apr baselines and compare \tool against them on \dfj 1.2 and 2.0. Our evaluation setting is on perfect fault localization -- where the ground-truth location of the bug is given to the \apr tool. We note that this is the preferred evaluation setting as it eliminates any differences caused by different fault localization methods~\cite{lutellier2020coconut, jiang2021cure, zhu2021recoder, tufano2018empstudy}. We follow the convention used in prior work~\cite{zhu2021recoder, jiang2021cure, liu2019tbar, ghanbari2019prapr} and directly report the bug fix results obtained from previous studies~\cite{xia2022alpharepair, ghanbari2019prapr, ye2022rewardrepair}. 

\subsection{Evaluation Metrics} 
\begin{itemize}
\item\textbf{Plausible patches} are patches that pass all test cases but may violate
the real user intent.
\item\textbf{Correct patches} are patches that are semantically equivalent to the developer patch.
Following common \apr practice, we determine semantic equivalency by manually examining each plausible patch.
\item\textbf{Patch compilation rate} is also used in many deep learning based \apr works~\cite{jiang2021cure, ye2022rewardrepair}, which indicates the percentage of compilable patches in all generated patches.
\end{itemize}

\section{Result Analysis}
\label{sec:result}

\subsection{RQ1: Comparison with Existing Tools}
\label{subsec:resultComparison}

In RQ1 and RQ2, we follow the prior approach for cloze-style \apr~\cite{xia2022alpharepair} to make use of repair templates for a faithful evaluation. Instead of replacing the entire buggy line with model-generated code, these templates systematically keep parts of the buggy line
to reduce the amount of code the \llm needs to generate. Note that we do not apply any repair templates in RQ3 and RQ4 because we consider a smaller number of samples there (i.e., 500 samples as shown in Section~\ref{sec:impl}) and also want to focus on the impact of different experimental configurations.

\begin{table}
\small
 \caption{Number of correct fixes on \dfj 1.2 single-hunk and \dfj 2.0 single-line bugs}
        \centering
        \label{tab:dfj12_evaluation}
        \begin{booktabs}{colspec={@{}llrrr@{}}}
        \toprule
        \SetCell[r=2]{m} Tool &\SetCell[r=2]{m} Methodology & \SetCell[c=3]{c}\#Correct Fixes\\
        \cmidrule[lr]{3-5}
        & & {\dfj 1.2} & {\dfj 2.0} & {Total}\\
        \midrule
        \coconut &\nmt{} & 30& - & -\\
        \dlfix &\nmt{} & 32 & - & -\\
        \prapr &Template & 35 & - & -\\
        \tbar &Template & 41 & 7 & 48\\
        \cure &\nmt{} & 43 & - & -\\
        \rewardrepair &\nmt{} & 45 & 24 & 69\\
        \recoder &\nmt{} & 51 & 10 & 61\\
        \alpharepair &\llm & 52 & 34 & 86\\
        \midrule[dotted]
        \tool &\llm & \textbf{66} & \textbf{50} & \textbf{116}\\
        \bottomrule
    \end{booktabs}
\end{table}

\begin{figure}
\centering
\includegraphics[width=\linewidth]{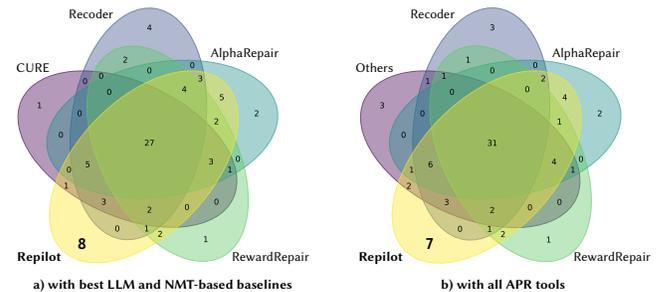}
\caption{Correct fix Venn diagrams on \dfj 1.2}
\label{fig:venn_dfj12}
\end{figure}

\parabf{\dfj 1.2.} We first compare \tool against the state-of-the-art \apr tools on single-hunk bugs from \dfj 1.2. \Cref{tab:dfj12_evaluation} shows the number of correct patches produced by \tool, evaluated in cloze-style, along with the baselines. \emph{\tool achieves the new state-of-the-art result of 66 correct bug fixes on \dfj 1.2, outperforming all previous \apr tools}. Figure~\ref{fig:venn_dfj12}a shows the Venn diagram of the unique bugs fixed for the top performing \llm- and \nmt{}-based
\apr tools where \tool is able to obtain the highest number of 8 unique bugs
Furthermore, Figure~\ref{fig:venn_dfj12}b compares the unique bugs fixed for all top-performing baselines and with all other \apr tools combined (Others). We observe that \tool is able to fix 7 bugs
that no other baselines have been able to fix so far.

\begin{figure}
\centering
\includegraphics[width=0.9\linewidth]{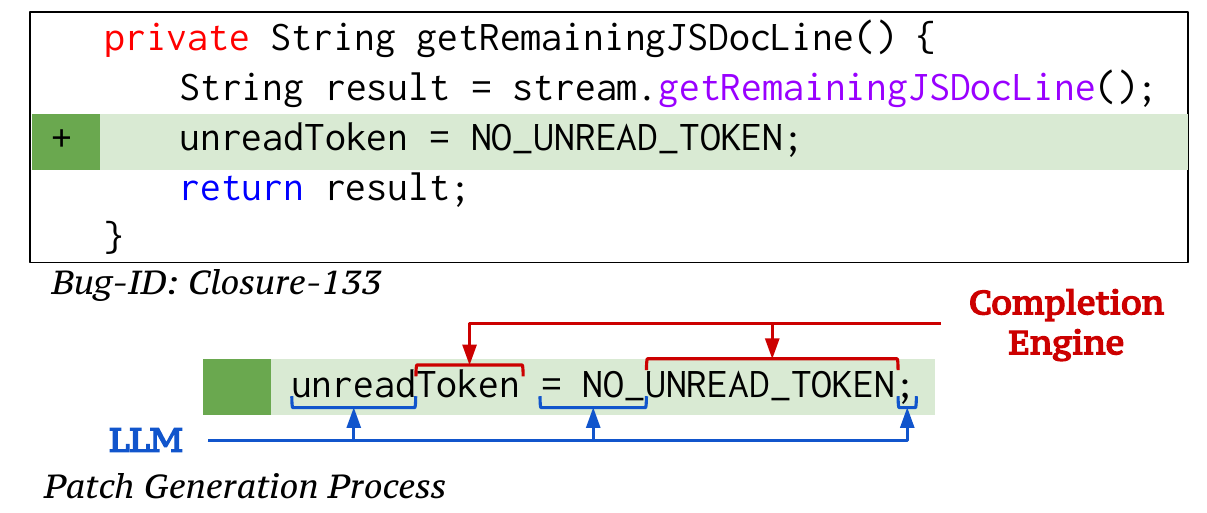}
\caption{Unique bug fix by \tool on \dfj 1.2}
\label{fig:unique_fix_dfj12}
\end{figure}

To demonstrate the ability of \tool to fix difficult bugs, Figure~\ref{fig:unique_fix_dfj12} shows a unique bug (\CodeIn{Closure-133}) from \dfj 1.2 that only \tool can fix. This bug is fixed by adding the new assignment statement using the global variable \CodeIn{NO\_UNREAD\_TOKEN} which is difficult to generate as it does not appear within the surrounding context of the bug location. \tool first uses \codetf to generate the initial prefix of \CodeIn{unread}. Then using the \cefull, \tool recognizes that \CodeIn{Token} is the only semantically correct continuation and directly performs active completion to return \CodeIn{unreadToken}. Similarly for generating \CodeIn{NO\_UNREAD\_TOKEN}, \tool first generates \CodeIn{NO\_} and then uses active completion to directly generate this rare identifier without having to repeatedly sample the \llm. It is difficult for prior \llm{-} and \nmt{}-based \apr tools to generate this fix as \llm{s} or \nmt models may not be able to complete this rare identifier since it requires multiple continuous steps to generate. In contrast, \tool, through the use of active completion, can directly generate this rare identifier given only the initial identifier prefix to quickly arrive at this correct patch. 

\parabf{\dfj 2.0.} We further evaluate \tool against baselines evaluated on the single-line bugs in \dfj 2.0. For these bugs, we follow prior approach for cloze-style \apr~\cite{xia2022alpharepair} to make use of repair templates. Instead of replacing the entire buggy line with model-generated code, these templates systematically keep parts of the buggy line (e.g., prefix or suffix, method parameters and calls) to reduce the amount of code the \llm needs to generate. We apply these repair templates for \dfj 2.0 single-line bugs only since they are designed for single-line bugs. \Cref{tab:dfj12_evaluation} also shows the number of correct fixes on \dfj 2.0 compared with the baselines. We observe that \emph{\tool is able to fix the highest number of bugs 50 (16 more than the next best baseline) on \dfj 2.0.} This improvement over existing baselines shows that \tool can generalize to two versions of \dfj datasets and demonstrates the power of repair templates to boost the performance of \llm-based \apr tools. 

Figure~\ref{fig:unique_fix_dfj2} shows a unique bug from \dfj 2.0 that only \tool can fix. First, \tool generates the patch up to the caret position. The \cefull then captures the exact type of the object from \CodeIn{Token.EndTag} to \CodeIn{String}. Using this information, \tool{} correctly prunes tokens that are not a part of the \CodeIn{String} class (e.g., \CodeIn{name} and \CodeIn{text}). Hence, the generated patch contains a valid \CodeIn{String} class method of \CodeIn{toLowerCase()} which correctly fixes this bug. Similar to the previous unique bug fix in \dfj 1.2, prior \llm-based \apr tools may waste a lot of time generating semantically incorrect continuations as they do not have access to the type information. Furthermore, \nmt{}-based \apr tools such as \cure~\cite{jiang2021cure}, over-approximating the list of valid identifiers by statically grabbing all the accessible fields, may not generate this fix since a pruned identifier (e.g., \CodeIn{name}) can also be valid for a different object type. \tool uses the \cefull to analyze partial programs and realize complex type propagation for effective pruning.

\begin{figure}
\centering
\includegraphics[width=0.9\linewidth]{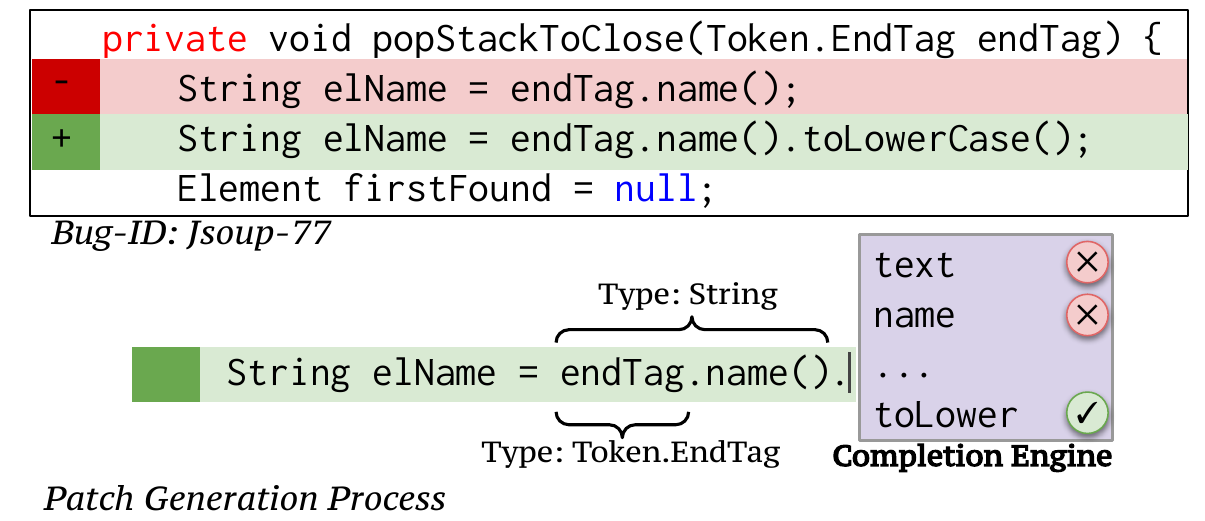}
\caption{Unique bug fix by \tool on \dfj 2.0}
\label{fig:unique_fix_dfj2}
\end{figure}

\subsection{RQ2: Compilation Rate Analysis}
\label{subsec:patchEfficiency}

\begin{table}
 \small
 \caption{Comparison with existing \apr tools on compilation rate on \dfj 1.2. "-" denotes data not available.}
        \centering
        \label{tab:compilation_rates}
        \begin{booktabs}{
            colspec={@{}lrrrr@{}},
            cell{Z}{2-Z} = {font=\bfseries},
            cell{1}{1} = {r=2}{m},
            cell{1}{2} = {c=4}{c}
        }
        \toprule
        Tool & \% Compilable Patches\\
        \cmidrule[lr]{2-5}
        & Top-30 & Top-100 & Top-1000 & Top-5000 \\
        \midrule
        \sequencer & 33\% & - & - & - \\
        \coconut & 24\% & 15\% & 6\% & 3\% \\
        \cure & 39\% & 28\% & 14\% & 9\% \\ 
        \alpharepair & 25\% & 22\% & 16\% & 13\% \\
        \rewardrepair & 45\% & 38\% & 33\%\(^1\) & - \\
        \midrule[dotted]
        \tool & 66\% & 62\% & 58\% & 59\% \\
        \bottomrule
    \end{booktabs}
    \begin{tablenotes}
    \item[*] \({}^{1}\) {\footnotesize This is the top 200 rate for \rewardrepair as it does not include top 1000}
    \end{tablenotes}
\end{table}

We evaluate the compilation rate of the patches generated by \tool compared with prior learning-based \apr techniques. Table~\ref{tab:compilation_rates} shows the percentage of compilable patches on the \dfj 1.2 dataset. We observe that across all numbers of patches generated, \tool significantly improves the percentage of compilable patches compared with prior tools. We first notice that \llm-based \apr tools (\tool and \alpharepair), are able to sustain their compilation rate compared with \nmt{}-based tools (\coconut and \cure) where the compilation rate drastically decreases as we increase the number of patches. This shows the ability for \llm{s} to generate large amounts of reasonable patches. \tool is able to sustain a near 60\% compilation percentage at 1000 patches generated while the prior approach is barely above 30\%. 

Compared with \cure~\cite{jiang2021cure}, where an overestimation of valid identifiers is obtained via static analysis and used to prune invalid tokens generated by \nmt model, \tool leverages the powerful \cefull to keep track of the current context to obtain a more accurate pruning step. Furthermore, compared with \rewardrepair~\cite{ye2022rewardrepair}, where the compilation rate is boosted through penalizing uncompilable patches during training, \tool directly uses a \llm combined with a \cefull to avoid this high cost of training a new model. Additionally, \tool uses the active completion ability of \cefull to directly generate these rare identifiers to further boost the compilation rate. As such, \tool is able to achieve the highest percentage of compilable patches
across all four different settings.

\subsection{RQ3: Ablation Study}
\label{subsec:ablation}
\begin{table}
\small
 \caption{Component contribution of \tool}
        \centering
        \label{tab:ablation}
        \begin{booktabs}[expand=\changescolor]{
            colsep=2.5pt,
            colspec={@{}lrrrrr@{}},
            column{Z}={fg=\changescolor},
            cell{Z}{2-Z} = {font=\bfseries}
        }
        \toprule
        Variant & {Generation\\Time} & {\%Compilable\\Patches} & {\%Plausible\\Patches} & {\#Plausible\\Fixes} & {\#Correct\\Fixes}\\
        \midrule
        \abvanilla{} & 0.232s & 43.2\% & 3.95\% & 56 & 37 \\
        \midrule[dotted]
        \abnomem{}  & 0.294s & 60.7\% & 5.02\% & 62 & 41 \\
        \abmem{} & 0.255s & 58.7\% & 4.82\% & 60 & 40 \\
        \abactive{} & 0.248s & 63.4\% & 5.21\% & 63 & 42 \\
        \bottomrule
    \end{booktabs}
\end{table}

To study the contribution of each component of \tool{} to its overall effectiveness,
we conduct an ablation study that aims at justifying the following hypothesis:
\begin{itemize}
    \item \Cref{algo:prune} (\fmPruneDecode) helps \llm
    to achieve valid (compilable) patches more efficiently on a pruned search space.
    \item Memorization (\Cref{subsec:memorize}) reduces the frequency of
    querying the \cefull, thus speeding up patch synthesis.
    \item Active completion provides further guidance of synthesis that and helps \tool{} efficiently achieve more valid patches.
    \item The plausible rate of patches becomes higher along with the compilation rate.
\end{itemize}

To give grounds for these hypotheses, we
set up the following four variants:
\begin{itemize}
\item \abvanilla{} uses only the base \llm (\codetf) for patch synthesis. 
\item \abnomem{} applies pruning defined in \Cref{algo:prune}.
\item \abmem{} leverages memorization (\Cref{subsec:memorize}) on top of pruning.
\item \abactive{} employs active completion for further guidance.
\end{itemize}
and evaluate them by comparing them against their efficiency in generating compilable, plausible patches, and correct patches.

\Cref{tab:ablation} shows the generation time (in seconds per patch), the contribution in terms of the percentage of compilable and plausible patches among all uniquely generated patches, the number of plausible fixes, and the number of correct fixes for each of the four variants on \dfj 1.2 single-hunk bugs.
We first observe that just using the base \llm for \apr (\abvanilla{}), we achieve the lowest compilation rate at 43.2\%. By adding the pruning provided by the \cefull, we can significantly improve the compilation rate to 60.7\%, the number of plausible fixes from 56 to 62, and the number of correct fixes from 37 to 41. Additional improvement is made by adding the active completion technique to achieve the full \tool with the highest compilation rate at 63.4\%, plausible percentage 5.21\%, the most number of plausible fixes at 63, and the most correct fixes at 42.

Looking at the patch generation time, starting from \abvanilla{}, adding pruning via \cefull incurs an over 25\% overhead. However, this can be significantly reduced by using memorization (\abnomem{}) to achieve around 10\% overhead by avoiding querying the \cefull once we know an identifier is invalid. Furthermore, active completion can further reduce the overhead to 7\% since instead of having to sample the \llm for each step in the generation, we can actively complete an identifier.

As a result, all the components contribute to the overall effectiveness of \tool. \tool can consistently increase the compilation and plausible rate, as well as produce more plausible/correct fixes while incurring minimal overhead compared with directly using \llm{s} for patch synthesis.

\color{\changescolor}
\subsection{RQ4: Generalizability}
\label{subsec:generalizability}

\begin{table*}[t]
 \small
    \caption{Generalizability of \tool across both subjects of bugs and models}
        \centering
        \label{tab:generalizability}
        \begin{booktabs}[expand={\changescolor}]{
          colspec={@{}lllrrrrr@{}},
          cell{3,5,7,9}{5-8} = {font=\bfseries},
          cell{2,4,6,7,8}{4} = {font=\bfseries},
        }
        \toprule[\changescolor]
        Variant & Model & Subject of Bugs & Generation Time & \%Compilable Patches & \%Plausible Patches & \#Plausible Fixes & \#Correct Fixes \\
        \midrule[\changescolor]
        \abvanilla{} & \codetf{}-large & \dfj 1.2 & 0.232s & 43.2\% & 3.95\% & 56 & 37 \\
        \abactive{} & \codetf{}-large & \dfj 1.2 & 0.248s & 63.4\% & 5.21\% & 63 & 42 \\
        \midrule[dotted, \changescolor]
        \abvanilla{} & \codetf{}-large & \dfj 2.0 & 0.230s & 46.7\% & 9.02\% & 59 & 41 \\
        \abactive{} & \codetf{}-large & \dfj 2.0 & 0.247s & 64.8\% & 12.02\% & 65 & 45 \\
        \midrule[\changescolor]
        \abvanilla{} & \incoder{}-6.7B & \dfj 1.2 & 1.70s & 32.4\% & 3.85\% & 70 & 48\\
        \abactive{} & \incoder{}-6.7B & \dfj 1.2 & 1.70s & 47.2\% & 4.96\% & 78 & 54 \\
        \midrule[dotted, \changescolor]
        \abvanilla{} & \incoder{}-6.7B & \dfj 2.0 & 1.67s & 34.6\% & 5.06\% & 67 & 45\\
        \abactive{} & \incoder{}-6.7B & \dfj 2.0 & 1.69s & 48.0\% & 6.87\% & 68 & 46 \\
        \bottomrule[\changescolor]
    \end{booktabs}
\end{table*}
\color{black}

To demonstrate the generalizability of \tool across different subjects of bugs and models, on the one hand, we further evaluate \tool with \codetf on all single-hunk bugs of \dfj 2.0. On the other hand, we additionally instantiate and evaluate \tool with a larger \incoder{}-6.7B model. Identical to RQ3, we also conduct 500 samples in RQ4 due to the high cost of \apr.

\Cref{tab:generalizability} shows the comparison between the baseline \abvanilla{} and our full \tool approach across different subjects of bugs and models. We consider the same set of \dfj 1.2 single-hunk bugs as in RQ3 and an extra set of \dfj 2.0 single-hunk bugs.

Upon investigation, we can see that \tool with \codetf surpasses the baseline on \dfj 1.2 as illustrated in RQ3. Furthermore, on \dfj 2.0, it can also achieve 18.1 percentage points (pp)
more compilable and 3.0 pp more plausible patches, as well as 6 more plausible fixes and 4 more correct fixes, with a 7.4\% overhead.

Meanwhile, when \tool is instantiated with \incoder, it still produces more compilable and plausible patches, as well as more plausible and correct fixes on both \dfj 1.2 and \dfj 2.0 over the baseline \incoder{}. It eventually gives 6 more correct fixes on \dfj 1.2 and 1 more on \dfj 2.0.

One major difference comparing \tool with \incoder and \codetf is that when \tool is equipped with \incoder, a much larger model than \codetf, it incurs negligible overhead. 
This is because compared to the high cost of autoregressive sampling using larger models, the extra cost from querying the \cefull is much smaller and thus trivializes the overhead of \tool when applied on larger models.
Also, the larger \incoder{} model, whether or not it is applied with \tool{}, can consistently fix more bugs across both \dfj 1.2 and 2.0 than \codetf{}, further confirming prior finding that larger \llm{s} often perform better for \apr~\cite{xia2022repairstudy}.

Overall, the experimental results indicate that \tool can generalize to different sets of bugs (both single-hunk bugs in \dfj 1.2 and 2.0) as well as larger \llm{s} (\incoder{})

\color{\changescolor}\section{Limitations}\color{black}
\label{sec:limitations}

First, to bring out \tool{}'s full potential, it is important that the \cefull can provide useful guidance while remaining strict (\Cref{def:strictCe}). However, it is generally more difficult to balance the usefulness and strictness of a \cefull in many dynamically typed programming languages, such as Python, compared with Java studied in this paper, which is a statically typed programming language. Meanwhile, there is a growing trend of dynamically typed languages adopting support for type hints~\cite{python, typescript, typedclojure}. Considering this, we believe that \tool can still provide significant advantages in such environments.

Another limitation of \tool{} lies in the evaluation.
On the one hand, while it is
true that an increase in the compilation rate of \tool can lead to the discovery of more plausible and correct fixes, it is important to note that a significantly higher compilation rate does not necessarily translate to a proportionally large increase in plausible and correct fixes.
On the other hand, \tool is only evaluated with \codetf for RQ1 and RQ2 with a 5000 sampling budget. \codetf is a rather ``small'' \llm compared to those \llm{s} with billions of parameters. Although we further include \incoder-6.7B as a multi-billion-parameter \llm in RQ4, due to time cost, we only sample 500 times per bug, which may be insufficient to reflect the distribution of the generated patches. Overall, the scope of our evaluation considering two \llm{s} (\codetf and \incoder) and one programming language (Java) is still narrow given that \tool is a general framework that can be instantiated with any pair of an \llm and a \cefull for some programming language.

Finally, despite the examples we show in the paper, our evaluation lacks strong empirical evidence to support the claim that \llm{s} have difficulty in generating rare tokens and how \tool solves the problem. Besides, our evaluation limits the application of \tool to patch synthesis, even though we claim that \tool can be applied to other code generation tasks. In the future, we will apply and evaluate \tool on more diverse code generation tasks.

\section{Threats to Validity}
\label{sec:threats}

\parabf{Internal.} We share the same main internal threat to validity with prior \apr tools where we have to manually examine each plausible patch to determine patch correctness. We address this by carefully analyzing each patch to determine if it is semantically equivalent to the reference developer patch. Furthermore, we have released our full set of correct patches for public evaluation~\cite{correctPatchandDataset}. 

Our use of the \codetf model poses another internal threat where the open-source training dataset of Git\-Hub projects~\cite{husain2020codesearchnet} may overlap with our evaluation of \dfj. We follow prior work~\cite{xia2022alpharepair, xia2022repairstudy} and address this by computing the correct bug fixes of \tool from \dfj that is part of the \codetf training data. In total, 7 out of 66 and 6 out of 50 overlap with training data on \dfj 1.2 and 2.0 respectively. For comparison fairness, if we were to exclude these 7 and 6 bugs and compare them with the previous baseline tools on the remaining bugs, we are still able to achieve the highest bug fixes at 59 and 44 (best baseline at 45 and 29). 
The same threat applies to the use of \incoder, but since its detailed training data is not revealed, we are unable to explicitly address this problem. To mitigate the problem, we only evaluate \incoder
in RQ4, where all the variants face the same potential leakage.

Moreover, our modified implementation of the completion engine requires manual inspection
to guarantee soundness property.
In practice, this is a significant trust base that may introduce false positives during pruning.
However, our theorem still provides a partial guarantee and is able to explain unsoundness.
At the same time, our evaluation result justifies our claims and demonstrates the practicality 
of \tool.

Finally, in our evaluation, we follow the convention used in prior work to directly report the bug fix results without reproducing them, which poses a threat to the reliability of the results. Meanwhile, we only run each of our experiments once, which could introduce extra statistical biases.

\parabf{External.} The main external threat to validity comes from our evaluation dataset where the performance of \tool may not generalize to other datasets. To address this, we compare \tool against state-of-the-art baselines on both \dfj 1.2 and 2.0 to show that the performance is sustained across both versions. To address this further, we plan to evaluate \tool on additional \apr datasets also across different programming languages. 
\section{Conclusion} 

We propose \tool{} --- the first \apr approach to combining the direct usage of \llm{s} (e.g., \codetf and \incoder) with on-the-fly guidance provided by \cefull{s}. During autoregressive token generation, \tool queries the \cefull not only to \emph{prune} invalid tokens but also to \emph{proactively complete} the currently generated partial program, thereby reducing the search space of the \llm.
Our evaluation on a subset of the widely-studied \dfj 1.2 and 2.0 datasets shows \tool is able to achieve the state-of-the-art results. Furthermore, \tool, through the usage of \cefull, is able to generate more valid and compilable patches than prior tools with minimal overhead compared with directly using \llm{s} for \apr.

\section*{Data Availability}
We have open-sourced Repilot, which can be accessed on GitHub at \url{https://github.com/ise-uiuc/Repilot}. Additionally, an immutable artifact for Repilot is publicly available on Zenodo~\cite{correctPatchandDataset}.

\section*{Acknowledgments}
We thank all the reviewers for their insightful comments. We also thank Yifeng Ding for his helpful discussion on this work. This work was partially supported by NSF grants \mbox{CCF-2131943} and \mbox{CCF-2141474}, as well as Kwai Inc.

\bibliographystyle{ACM-Reference-Format}
\bibliography{reference}


\begin{thebibliography}{74}


\ifx \showCODEN    \undefined \def \showCODEN     #1{\unskip}     \fi
\ifx \showDOI      \undefined \def \showDOI       #1{#1}\fi
\ifx \showISBNx    \undefined \def \showISBNx     #1{\unskip}     \fi
\ifx \showISBNxiii \undefined \def \showISBNxiii  #1{\unskip}     \fi
\ifx \showISSN     \undefined \def \showISSN      #1{\unskip}     \fi
\ifx \showLCCN     \undefined \def \showLCCN      #1{\unskip}     \fi
\ifx \shownote     \undefined \def \shownote      #1{#1}          \fi
\ifx \showarticletitle \undefined \def \showarticletitle #1{#1}   \fi
\ifx \showURL      \undefined \def \showURL       {\relax}        \fi
\providecommand\bibfield[2]{#2}
\providecommand\bibinfo[2]{#2}
\providecommand\natexlab[1]{#1}
\providecommand\showeprint[2][]{arXiv:#2}

\bibitem[ecl(2023)]%
        {eclipseLS}
 \bibinfo{year}{2023}\natexlab{}.
\newblock \bibinfo{title}{Eclipse JDT LS}.
\newblock
\newblock
\newblock
\shownote{\url{https://projects.eclipse.org/projects/eclipse.jdt.ls}}.


\bibitem[Aghajanyan et~al\mbox{.}(2022)]%
        {aghajanyan2022cm3}
\bibfield{author}{\bibinfo{person}{Armen Aghajanyan}, \bibinfo{person}{Bernie
  Huang}, \bibinfo{person}{Candace Ross}, \bibinfo{person}{Vladimir Karpukhin},
  \bibinfo{person}{Hu Xu}, \bibinfo{person}{Naman Goyal},
  \bibinfo{person}{Dmytro Okhonko}, \bibinfo{person}{Mandar Joshi},
  \bibinfo{person}{Gargi Ghosh}, \bibinfo{person}{Mike Lewis}, {and}
  \bibinfo{person}{Luke Zettlemoyer}.} \bibinfo{year}{2022}\natexlab{}.
\newblock \showarticletitle{{CM3:} {A} Causal Masked Multimodal Model of the
  Internet}.
\newblock \bibinfo{journal}{\emph{CoRR}}  \bibinfo{volume}{abs/2201.07520}
  (\bibinfo{year}{2022}).
\newblock
\showeprint[arXiv]{2201.07520}
\urldef\tempurl%
\url{https://arxiv.org/abs/2201.07520}
\showURL{%
\tempurl}


\bibitem[Ahmad et~al\mbox{.}(2021)]%
        {ahmad2021PLBART}
\bibfield{author}{\bibinfo{person}{Wasi~Uddin Ahmad}, \bibinfo{person}{Saikat
  Chakraborty}, \bibinfo{person}{Baishakhi Ray}, {and} \bibinfo{person}{Kai-Wei
  Chang}.} \bibinfo{year}{2021}\natexlab{}.
\newblock \bibinfo{title}{Unified Pre-training for Program Understanding and
  Generation}.
\newblock
\newblock
\showeprint[arxiv]{2103.06333}~[cs.CL]


\bibitem[Austin et~al\mbox{.}(2021)]%
        {austin2021synthesis}
\bibfield{author}{\bibinfo{person}{Jacob Austin}, \bibinfo{person}{Augustus
  Odena}, \bibinfo{person}{Maxwell~I. Nye}, \bibinfo{person}{Maarten Bosma},
  \bibinfo{person}{Henryk Michalewski}, \bibinfo{person}{David Dohan},
  \bibinfo{person}{Ellen Jiang}, \bibinfo{person}{Carrie~J. Cai},
  \bibinfo{person}{Michael Terry}, \bibinfo{person}{Quoc~V. Le}, {and}
  \bibinfo{person}{Charles Sutton}.} \bibinfo{year}{2021}\natexlab{}.
\newblock \showarticletitle{Program Synthesis with Large Language Models}.
\newblock \bibinfo{journal}{\emph{CoRR}}  \bibinfo{volume}{abs/2108.07732}
  (\bibinfo{year}{2021}).
\newblock
\showeprint[arXiv]{2108.07732}
\urldef\tempurl%
\url{https://arxiv.org/abs/2108.07732}
\showURL{%
\tempurl}


\bibitem[Barke et~al\mbox{.}(2023)]%
        {groundcopilot}
\bibfield{author}{\bibinfo{person}{Shraddha Barke}, \bibinfo{person}{Michael~B.
  James}, {and} \bibinfo{person}{Nadia Polikarpova}.}
  \bibinfo{year}{2023}\natexlab{}.
\newblock \showarticletitle{Grounded Copilot: How Programmers Interact with
  Code-Generating Models}.
\newblock \bibinfo{journal}{\emph{Proc. ACM Program. Lang.}}
  \bibinfo{volume}{7}, \bibinfo{number}{OOPSLA1}, Article
  \bibinfo{articleno}{78} (\bibinfo{date}{apr} \bibinfo{year}{2023}),
  \bibinfo{numpages}{27}~pages.
\newblock
\urldef\tempurl%
\url{https://doi.org/10.1145/3586030}
\showDOI{\tempurl}


\bibitem[Barr et~al\mbox{.}(2014)]%
        {Barr14fse}
\bibfield{author}{\bibinfo{person}{Earl~T. Barr}, \bibinfo{person}{Yuriy Brun},
  \bibinfo{person}{Premkumar Devanbu}, \bibinfo{person}{Mark Harman}, {and}
  \bibinfo{person}{Federica Sarro}.} \bibinfo{year}{2014}\natexlab{}.
\newblock \showarticletitle{The Plastic Surgery Hypothesis}. In
  \bibinfo{booktitle}{\emph{Proceedings of the 22nd ACM SIGSOFT International
  Symposium on Foundations of Software Engineering}} (Hong Kong, China)
  \emph{(\bibinfo{series}{FSE 2014})}. \bibinfo{publisher}{Association for
  Computing Machinery}, \bibinfo{address}{New York, NY, USA},
  \bibinfo{pages}{306–317}.
\newblock
\showISBNx{9781450330565}
\urldef\tempurl%
\url{https://doi.org/10.1145/2635868.2635898}
\showDOI{\tempurl}


\bibitem[Brown et~al\mbox{.}(2020)]%
        {brown2020gpt3}
\bibfield{author}{\bibinfo{person}{Tom~B. Brown}, \bibinfo{person}{Benjamin
  Mann}, \bibinfo{person}{Nick Ryder}, \bibinfo{person}{Melanie Subbiah},
  \bibinfo{person}{Jared Kaplan}, \bibinfo{person}{Prafulla Dhariwal},
  \bibinfo{person}{Arvind Neelakantan}, \bibinfo{person}{Pranav Shyam},
  \bibinfo{person}{Girish Sastry}, \bibinfo{person}{Amanda Askell},
  \bibinfo{person}{Sandhini Agarwal}, \bibinfo{person}{Ariel Herbert{-}Voss},
  \bibinfo{person}{Gretchen Krueger}, \bibinfo{person}{Tom Henighan},
  \bibinfo{person}{Rewon Child}, \bibinfo{person}{Aditya Ramesh},
  \bibinfo{person}{Daniel~M. Ziegler}, \bibinfo{person}{Jeffrey Wu},
  \bibinfo{person}{Clemens Winter}, \bibinfo{person}{Christopher Hesse},
  \bibinfo{person}{Mark Chen}, \bibinfo{person}{Eric Sigler},
  \bibinfo{person}{Mateusz Litwin}, \bibinfo{person}{Scott Gray},
  \bibinfo{person}{Benjamin Chess}, \bibinfo{person}{Jack Clark},
  \bibinfo{person}{Christopher Berner}, \bibinfo{person}{Sam McCandlish},
  \bibinfo{person}{Alec Radford}, \bibinfo{person}{Ilya Sutskever}, {and}
  \bibinfo{person}{Dario Amodei}.} \bibinfo{year}{2020}\natexlab{}.
\newblock \showarticletitle{Language Models are Few-Shot Learners}.
\newblock \bibinfo{journal}{\emph{CoRR}}  \bibinfo{volume}{abs/2005.14165}
  (\bibinfo{year}{2020}).
\newblock
\showeprint[arXiv]{2005.14165}
\urldef\tempurl%
\url{https://arxiv.org/abs/2005.14165}
\showURL{%
\tempurl}


\bibitem[Cao et~al\mbox{.}(2023)]%
        {cao2023study}
\bibfield{author}{\bibinfo{person}{Jialun Cao}, \bibinfo{person}{Meiziniu Li},
  \bibinfo{person}{Ming Wen}, {and} \bibinfo{person}{Shing{-}Chi Cheung}.}
  \bibinfo{year}{2023}\natexlab{}.
\newblock \showarticletitle{A study on Prompt Design, Advantages and
  Limitations of ChatGPT for Deep Learning Program Repair}.
\newblock \bibinfo{journal}{\emph{CoRR}}  \bibinfo{volume}{abs/2304.08191}
  (\bibinfo{year}{2023}).
\newblock
\urldef\tempurl%
\url{https://doi.org/10.48550/ARXIV.2304.08191}
\showDOI{\tempurl}
\showeprint[arXiv]{2304.08191}


\bibitem[Chen et~al\mbox{.}(2017)]%
        {chen2017jaid}
\bibfield{author}{\bibinfo{person}{Liushan Chen}, \bibinfo{person}{Yu Pei},
  {and} \bibinfo{person}{Carlo~A. Furia}.} \bibinfo{year}{2017}\natexlab{}.
\newblock \showarticletitle{Contract-based program repair without the
  contracts}. In \bibinfo{booktitle}{\emph{2017 32nd IEEE/ACM International
  Conference on Automated Software Engineering (ASE)}}.
  \bibinfo{pages}{637--647}.
\newblock
\urldef\tempurl%
\url{https://doi.org/10.1109/ASE.2017.8115674}
\showDOI{\tempurl}


\bibitem[Chen et~al\mbox{.}(2021b)]%
        {chen2021codex}
\bibfield{author}{\bibinfo{person}{Mark Chen}, \bibinfo{person}{Jerry Tworek},
  \bibinfo{person}{Heewoo Jun}, \bibinfo{person}{Qiming Yuan},
  \bibinfo{person}{Henrique~Ponde de Oliveira~Pinto}, \bibinfo{person}{Jared
  Kaplan}, \bibinfo{person}{Harri Edwards}, \bibinfo{person}{Yuri Burda},
  \bibinfo{person}{Nicholas Joseph}, \bibinfo{person}{Greg Brockman},
  \bibinfo{person}{Alex Ray}, \bibinfo{person}{Raul Puri},
  \bibinfo{person}{Gretchen Krueger}, \bibinfo{person}{Michael Petrov},
  \bibinfo{person}{Heidy Khlaaf}, \bibinfo{person}{Girish Sastry},
  \bibinfo{person}{Pamela Mishkin}, \bibinfo{person}{Brooke Chan},
  \bibinfo{person}{Scott Gray}, \bibinfo{person}{Nick Ryder},
  \bibinfo{person}{Mikhail Pavlov}, \bibinfo{person}{Alethea Power},
  \bibinfo{person}{Lukasz Kaiser}, \bibinfo{person}{Mohammad Bavarian},
  \bibinfo{person}{Clemens Winter}, \bibinfo{person}{Philippe Tillet},
  \bibinfo{person}{Felipe~Petroski Such}, \bibinfo{person}{Dave Cummings},
  \bibinfo{person}{Matthias Plappert}, \bibinfo{person}{Fotios Chantzis},
  \bibinfo{person}{Elizabeth Barnes}, \bibinfo{person}{Ariel Herbert-Voss},
  \bibinfo{person}{William~Hebgen Guss}, \bibinfo{person}{Alex Nichol},
  \bibinfo{person}{Alex Paino}, \bibinfo{person}{Nikolas Tezak},
  \bibinfo{person}{Jie Tang}, \bibinfo{person}{Igor Babuschkin},
  \bibinfo{person}{Suchir Balaji}, \bibinfo{person}{Shantanu Jain},
  \bibinfo{person}{William Saunders}, \bibinfo{person}{Christopher Hesse},
  \bibinfo{person}{Andrew~N. Carr}, \bibinfo{person}{Jan Leike},
  \bibinfo{person}{Josh Achiam}, \bibinfo{person}{Vedant Misra},
  \bibinfo{person}{Evan Morikawa}, \bibinfo{person}{Alec Radford},
  \bibinfo{person}{Matthew Knight}, \bibinfo{person}{Miles Brundage},
  \bibinfo{person}{Mira Murati}, \bibinfo{person}{Katie Mayer},
  \bibinfo{person}{Peter Welinder}, \bibinfo{person}{Bob McGrew},
  \bibinfo{person}{Dario Amodei}, \bibinfo{person}{Sam McCandlish},
  \bibinfo{person}{Ilya Sutskever}, {and} \bibinfo{person}{Wojciech Zaremba}.}
  \bibinfo{year}{2021}\natexlab{b}.
\newblock \bibinfo{title}{Evaluating Large Language Models Trained on Code}.
\newblock
\newblock
\showeprint[arxiv]{2107.03374}~[cs.LG]


\bibitem[Chen et~al\mbox{.}(2021a)]%
        {chen2018sequencer}
\bibfield{author}{\bibinfo{person}{Zimin Chen}, \bibinfo{person}{Steve
  Kommrusch}, \bibinfo{person}{Michele Tufano}, \bibinfo{person}{Louis-Noël
  Pouchet}, \bibinfo{person}{Denys Poshyvanyk}, {and} \bibinfo{person}{Martin
  Monperrus}.} \bibinfo{year}{2021}\natexlab{a}.
\newblock \showarticletitle{SequenceR: Sequence-to-Sequence Learning for
  End-to-End Program Repair}.
\newblock \bibinfo{journal}{\emph{IEEE Transactions on Software Engineering}}
  \bibinfo{volume}{47}, \bibinfo{number}{9} (\bibinfo{year}{2021}),
  \bibinfo{pages}{1943--1959}.
\newblock
\urldef\tempurl%
\url{https://doi.org/10.1109/TSE.2019.2940179}
\showDOI{\tempurl}


\bibitem[Clojure(2023)]%
        {typedclojure}
Clojure \bibinfo{year}{2023}\natexlab{}.
\newblock \bibinfo{title}{Typed Clojure: An Optional Type System for Clojure}.
\newblock
\newblock
\newblock
\shownote{\url{https://typedclojure.org}}.


\bibitem[DeMarco et~al\mbox{.}(2014)]%
        {demacro2014nopol}
\bibfield{author}{\bibinfo{person}{Favio DeMarco}, \bibinfo{person}{Jifeng
  Xuan}, \bibinfo{person}{Daniel Le~Berre}, {and} \bibinfo{person}{Martin
  Monperrus}.} \bibinfo{year}{2014}\natexlab{}.
\newblock \showarticletitle{Automatic Repair of Buggy If Conditions and Missing
  Preconditions with SMT}. In \bibinfo{booktitle}{\emph{Proceedings of the 6th
  International Workshop on Constraints in Software Testing, Verification, and
  Analysis}} (Hyderabad, India) \emph{(\bibinfo{series}{CSTVA 2014})}.
  \bibinfo{publisher}{Association for Computing Machinery},
  \bibinfo{address}{New York, NY, USA}, \bibinfo{pages}{30–39}.
\newblock
\showISBNx{9781450328470}
\urldef\tempurl%
\url{https://doi.org/10.1145/2593735.2593740}
\showDOI{\tempurl}


\bibitem[Deng et~al\mbox{.}(2023)]%
        {titanfuzz}
\bibfield{author}{\bibinfo{person}{Yinlin Deng},
  \bibinfo{person}{Chunqiu~Steven Xia}, \bibinfo{person}{Haoran Peng},
  \bibinfo{person}{Chenyuan Yang}, {and} \bibinfo{person}{Lingming Zhang}.}
  \bibinfo{year}{2023}\natexlab{}.
\newblock \showarticletitle{Large Language Models Are Zero-Shot Fuzzers:
  Fuzzing Deep-Learning Libraries via Large Language Models}. In
  \bibinfo{booktitle}{\emph{Proceedings of the 32nd ACM SIGSOFT International
  Symposium on Software Testing and Analysis}} (Seattle, WA, USA)
  \emph{(\bibinfo{series}{ISSTA 2023})}. \bibinfo{publisher}{Association for
  Computing Machinery}, \bibinfo{address}{New York, NY, USA},
  \bibinfo{pages}{423–435}.
\newblock
\showISBNx{9798400702211}
\urldef\tempurl%
\url{https://doi.org/10.1145/3597926.3598067}
\showDOI{\tempurl}


\bibitem[Devlin et~al\mbox{.}(2019)]%
        {devlin2018bert}
\bibfield{author}{\bibinfo{person}{Jacob Devlin}, \bibinfo{person}{Ming-Wei
  Chang}, \bibinfo{person}{Kenton Lee}, {and} \bibinfo{person}{Kristina
  Toutanova}.} \bibinfo{year}{2019}\natexlab{}.
\newblock \showarticletitle{{BERT}: Pre-training of Deep Bidirectional
  Transformers for Language Understanding}. In
  \bibinfo{booktitle}{\emph{Proceedings of the 2019 Conference of the North
  {A}merican Chapter of the Association for Computational Linguistics: Human
  Language Technologies, Volume 1 (Long and Short Papers)}}.
  \bibinfo{publisher}{Association for Computational Linguistics},
  \bibinfo{address}{Minneapolis, Minnesota}, \bibinfo{pages}{4171--4186}.
\newblock
\urldef\tempurl%
\url{https://doi.org/10.18653/v1/N19-1423}
\showDOI{\tempurl}


\bibitem[Ding et~al\mbox{.}(2023)]%
        {ding2023crosscodeeval}
\bibfield{author}{\bibinfo{person}{Yangruibo Ding}, \bibinfo{person}{Zijian
  Wang}, \bibinfo{person}{Wasi~Uddin Ahmad}, \bibinfo{person}{Hantian Ding},
  \bibinfo{person}{Ming Tan}, \bibinfo{person}{Nihal Jain},
  \bibinfo{person}{Murali~Krishna Ramanathan}, \bibinfo{person}{Ramesh
  Nallapati}, \bibinfo{person}{Parminder Bhatia}, \bibinfo{person}{Dan Roth},
  {and} \bibinfo{person}{Bing Xiang}.} \bibinfo{year}{2023}\natexlab{}.
\newblock \bibinfo{title}{CrossCodeEval: A Diverse and Multilingual Benchmark
  for Cross-File Code Completion}.
\newblock
\newblock
\showeprint[arxiv]{2310.11248}~[cs.LG]


\bibitem[Feng et~al\mbox{.}(2020)]%
        {feng2020codebert}
\bibfield{author}{\bibinfo{person}{Zhangyin Feng}, \bibinfo{person}{Daya Guo},
  \bibinfo{person}{Duyu Tang}, \bibinfo{person}{Nan Duan},
  \bibinfo{person}{Xiaocheng Feng}, \bibinfo{person}{Ming Gong},
  \bibinfo{person}{Linjun Shou}, \bibinfo{person}{Bing Qin},
  \bibinfo{person}{Ting Liu}, \bibinfo{person}{Daxin Jiang}, {and}
  \bibinfo{person}{Ming Zhou}.} \bibinfo{year}{2020}\natexlab{}.
\newblock \showarticletitle{{C}ode{BERT}: A Pre-Trained Model for Programming
  and Natural Languages}.
\newblock \bibinfo{journal}{\emph{CoRR}}  \bibinfo{volume}{abs/2002.08155}.
\newblock
\showeprint[arXiv]{2002.08155}
\urldef\tempurl%
\url{https://arxiv.org/abs/2002.08155}
\showURL{%
\tempurl}


\bibitem[Foundation and Wei(2023)]%
        {modifiedJDT}
\bibfield{author}{\bibinfo{person}{Eclipse Foundation} {and}
  \bibinfo{person}{Yuxiang Wei}.} \bibinfo{year}{2023}\natexlab{}.
\newblock \bibinfo{booktitle}{\emph{{UniverseFly/eclipse.jdt.ls: Modified
  Eclipse JDT LS 1.0.3}}}.
\newblock
\urldef\tempurl%
\url{https://doi.org/10.5281/zenodo.8278193}
\showDOI{\tempurl}


\bibitem[Fried et~al\mbox{.}(2023)]%
        {fried2023incoder}
\bibfield{author}{\bibinfo{person}{Daniel Fried}, \bibinfo{person}{Armen
  Aghajanyan}, \bibinfo{person}{Jessy Lin}, \bibinfo{person}{Sida Wang},
  \bibinfo{person}{Eric Wallace}, \bibinfo{person}{Freda Shi},
  \bibinfo{person}{Ruiqi Zhong}, \bibinfo{person}{Scott Yih},
  \bibinfo{person}{Luke Zettlemoyer}, {and} \bibinfo{person}{Mike Lewis}.}
  \bibinfo{year}{2023}\natexlab{}.
\newblock \showarticletitle{InCoder: A Generative Model for Code Infilling and
  Synthesis}. In \bibinfo{booktitle}{\emph{The Eleventh International
  Conference on Learning Representations}}.
\newblock
\urldef\tempurl%
\url{https://openreview.net/forum?id=hQwb-lbM6EL}
\showURL{%
\tempurl}


\bibitem[Gazzola et~al\mbox{.}(2019)]%
        {gazzola2019aprsurvey}
\bibfield{author}{\bibinfo{person}{Luca Gazzola}, \bibinfo{person}{Daniela
  Micucci}, {and} \bibinfo{person}{Leonardo Mariani}.}
  \bibinfo{year}{2019}\natexlab{}.
\newblock \showarticletitle{Automatic Software Repair: A Survey}.
\newblock \bibinfo{journal}{\emph{IEEE Transactions on Software Engineering}}
  \bibinfo{volume}{45}, \bibinfo{number}{1} (\bibinfo{year}{2019}),
  \bibinfo{pages}{34--67}.
\newblock
\urldef\tempurl%
\url{https://doi.org/10.1109/TSE.2017.2755013}
\showDOI{\tempurl}


\bibitem[Ghanbari et~al\mbox{.}(2019)]%
        {ghanbari2019prapr}
\bibfield{author}{\bibinfo{person}{Ali Ghanbari}, \bibinfo{person}{Samuel
  Benton}, {and} \bibinfo{person}{Lingming Zhang}.}
  \bibinfo{year}{2019}\natexlab{}.
\newblock \showarticletitle{Practical Program Repair via Bytecode Mutation}. In
  \bibinfo{booktitle}{\emph{Proceedings of the 28th ACM SIGSOFT International
  Symposium on Software Testing and Analysis}} (Beijing, China)
  \emph{(\bibinfo{series}{ISSTA 2019})}. \bibinfo{publisher}{Association for
  Computing Machinery}, \bibinfo{address}{New York, NY, USA},
  \bibinfo{pages}{19–30}.
\newblock
\showISBNx{9781450362245}
\urldef\tempurl%
\url{https://doi.org/10.1145/3293882.3330559}
\showDOI{\tempurl}


\bibitem[GithubCopilot(2023)]%
        {GithubCopilot}
GithubCopilot \bibinfo{year}{2023}\natexlab{}.
\newblock \bibinfo{title}{GitHub Copilot: Your AI pair programmer}.
\newblock
\newblock
\newblock
\shownote{\url{https://github.com/features/copilot}}.


\bibitem[Guo et~al\mbox{.}(2021)]%
        {guo2021graphcodebert}
\bibfield{author}{\bibinfo{person}{Daya Guo}, \bibinfo{person}{Shuo Ren},
  \bibinfo{person}{Shuai Lu}, \bibinfo{person}{Zhangyin Feng},
  \bibinfo{person}{Duyu Tang}, \bibinfo{person}{Shujie LIU},
  \bibinfo{person}{Long Zhou}, \bibinfo{person}{Nan Duan},
  \bibinfo{person}{Alexey Svyatkovskiy}, \bibinfo{person}{Shengyu Fu},
  \bibinfo{person}{Michele Tufano}, \bibinfo{person}{Shao~Kun Deng},
  \bibinfo{person}{Colin Clement}, \bibinfo{person}{Dawn Drain},
  \bibinfo{person}{Neel Sundaresan}, \bibinfo{person}{Jian Yin},
  \bibinfo{person}{Daxin Jiang}, {and} \bibinfo{person}{Ming Zhou}.}
  \bibinfo{year}{2021}\natexlab{}.
\newblock \showarticletitle{GraphCodeBERT: Pre-training Code Representations
  with Data Flow}. In \bibinfo{booktitle}{\emph{International Conference on
  Learning Representations}}.
\newblock
\urldef\tempurl%
\url{https://openreview.net/forum?id=jLoC4ez43PZ}
\showURL{%
\tempurl}


\bibitem[Holtzman et~al\mbox{.}(2020)]%
        {holtzman2019nucleus}
\bibfield{author}{\bibinfo{person}{Ari Holtzman}, \bibinfo{person}{Jan Buys},
  \bibinfo{person}{Li Du}, \bibinfo{person}{Maxwell Forbes}, {and}
  \bibinfo{person}{Yejin Choi}.} \bibinfo{year}{2020}\natexlab{}.
\newblock \showarticletitle{The Curious Case of Neural Text Degeneration}. In
  \bibinfo{booktitle}{\emph{International Conference on Learning
  Representations}}.
\newblock
\urldef\tempurl%
\url{https://openreview.net/forum?id=rygGQyrFvH}
\showURL{%
\tempurl}


\bibitem[Hua et~al\mbox{.}(2018)]%
        {hua2018sketchfix}
\bibfield{author}{\bibinfo{person}{Jinru Hua}, \bibinfo{person}{Mengshi Zhang},
  \bibinfo{person}{Kaiyuan Wang}, {and} \bibinfo{person}{Sarfraz Khurshid}.}
  \bibinfo{year}{2018}\natexlab{}.
\newblock \showarticletitle{SketchFix: A Tool for Automated Program Repair
  Approach Using Lazy Candidate Generation}. In
  \bibinfo{booktitle}{\emph{Proceedings of the 2018 26th ACM Joint Meeting on
  European Software Engineering Conference and Symposium on the Foundations of
  Software Engineering}} (Lake Buena Vista, FL, USA)
  \emph{(\bibinfo{series}{ESEC/FSE 2018})}. \bibinfo{publisher}{Association for
  Computing Machinery}, \bibinfo{address}{New York, NY, USA},
  \bibinfo{pages}{888–891}.
\newblock
\showISBNx{9781450355735}
\urldef\tempurl%
\url{https://doi.org/10.1145/3236024.3264600}
\showDOI{\tempurl}


\bibitem[HuggingFace(2023)]%
        {HuggingFaceWebPage}
HuggingFace \bibinfo{year}{2023}\natexlab{}.
\newblock \bibinfo{title}{Hugging Face}.
\newblock
\newblock
\newblock
\shownote{\url{https://huggingface.co}}.


\bibitem[Husain et~al\mbox{.}(2019)]%
        {husain2020codesearchnet}
\bibfield{author}{\bibinfo{person}{Hamel Husain}, \bibinfo{person}{Ho{-}Hsiang
  Wu}, \bibinfo{person}{Tiferet Gazit}, \bibinfo{person}{Miltiadis Allamanis},
  {and} \bibinfo{person}{Marc Brockschmidt}.} \bibinfo{year}{2019}\natexlab{}.
\newblock \showarticletitle{CodeSearchNet Challenge: Evaluating the State of
  Semantic Code Search}.
\newblock \bibinfo{journal}{\emph{CoRR}}  \bibinfo{volume}{abs/1909.09436}
  (\bibinfo{year}{2019}).
\newblock
\showeprint[arXiv]{1909.09436}
\urldef\tempurl%
\url{http://arxiv.org/abs/1909.09436}
\showURL{%
\tempurl}


\bibitem[Jiang et~al\mbox{.}(2018)]%
        {jiang2018simfix}
\bibfield{author}{\bibinfo{person}{Jiajun Jiang}, \bibinfo{person}{Yingfei
  Xiong}, \bibinfo{person}{Hongyu Zhang}, \bibinfo{person}{Qing Gao}, {and}
  \bibinfo{person}{Xiangqun Chen}.} \bibinfo{year}{2018}\natexlab{}.
\newblock \showarticletitle{Shaping Program Repair Space with Existing Patches
  and Similar Code}. In \bibinfo{booktitle}{\emph{ISSTA 2018}} (Amsterdam,
  Netherlands). \bibinfo{publisher}{Association for Computing Machinery},
  \bibinfo{address}{New York, NY, USA}, \bibinfo{pages}{298–309}.
\newblock
\showISBNx{9781450356992}
\urldef\tempurl%
\url{https://doi.org/10.1145/3213846.3213871}
\showDOI{\tempurl}


\bibitem[Jiang et~al\mbox{.}(2021b)]%
        {jiang2021cure}
\bibfield{author}{\bibinfo{person}{Nan Jiang}, \bibinfo{person}{Thibaud
  Lutellier}, {and} \bibinfo{person}{Lin Tan}.}
  \bibinfo{year}{2021}\natexlab{b}.
\newblock \showarticletitle{CURE: Code-Aware Neural Machine Translation for
  Automatic Program Repair}. In \bibinfo{booktitle}{\emph{Proceedings of the
  43rd International Conference on Software Engineering}} (Madrid, Spain)
  \emph{(\bibinfo{series}{ICSE '21})}. \bibinfo{publisher}{IEEE Press},
  \bibinfo{pages}{1161–1173}.
\newblock
\showISBNx{9781450390859}
\urldef\tempurl%
\url{https://doi.org/10.1109/ICSE43902.2021.00107}
\showDOI{\tempurl}


\bibitem[Jiang et~al\mbox{.}(2021a)]%
        {jiang2021extract}
\bibfield{author}{\bibinfo{person}{Yanjie Jiang}, \bibinfo{person}{Hui Liu},
  \bibinfo{person}{Nan Niu}, \bibinfo{person}{Lu Zhang}, {and}
  \bibinfo{person}{Yamin Hu}.} \bibinfo{year}{2021}\natexlab{a}.
\newblock \showarticletitle{Extracting Concise Bug-Fixing Patches from
  Human-Written Patches in Version Control Systems}. In
  \bibinfo{booktitle}{\emph{Proceedings of the 43rd International Conference on
  Software Engineering}} (Madrid, Spain) \emph{(\bibinfo{series}{ICSE '21})}.
  \bibinfo{publisher}{IEEE Press}, \bibinfo{pages}{686–698}.
\newblock
\showISBNx{9781450390859}
\urldef\tempurl%
\url{https://doi.org/10.1109/ICSE43902.2021.00069}
\showDOI{\tempurl}


\bibitem[Joshi et~al\mbox{.}(2023)]%
        {joshi2023repair}
\bibfield{author}{\bibinfo{person}{Harshit Joshi}, \bibinfo{person}{José
  Cambronero}, \bibinfo{person}{Sumit Gulwani}, \bibinfo{person}{Vu Le},
  \bibinfo{person}{Ivan Radicek}, {and} \bibinfo{person}{Gust Verbruggen}.}
  \bibinfo{year}{2023}\natexlab{}.
\newblock \bibinfo{booktitle}{\emph{Repair Is Nearly Generation: Multilingual
  Program Repair with LLMs}}.
\newblock AAAI.
\newblock
\urldef\tempurl%
\url{https://www.microsoft.com/en-us/research/publication/repair-is-nearly-generation-multilingual-program-repair-with-llms/}
\showURL{%
\tempurl}


\bibitem[Just et~al\mbox{.}(2014)]%
        {just2014dfj}
\bibfield{author}{\bibinfo{person}{Ren\'{e} Just}, \bibinfo{person}{Darioush
  Jalali}, {and} \bibinfo{person}{Michael~D. Ernst}.}
  \bibinfo{year}{2014}\natexlab{}.
\newblock \showarticletitle{Defects4J: A Database of Existing Faults to Enable
  Controlled Testing Studies for Java Programs}. In
  \bibinfo{booktitle}{\emph{Proceedings of the 2014 International Symposium on
  Software Testing and Analysis}} (San Jose, CA, USA)
  \emph{(\bibinfo{series}{ISSTA 2014})}. \bibinfo{publisher}{Association for
  Computing Machinery}, \bibinfo{address}{New York, NY, USA},
  \bibinfo{pages}{437–440}.
\newblock
\showISBNx{9781450326452}
\urldef\tempurl%
\url{https://doi.org/10.1145/2610384.2628055}
\showDOI{\tempurl}


\bibitem[Kolak et~al\mbox{.}(2022)]%
        {kolak2022patch}
\bibfield{author}{\bibinfo{person}{Sophia~D Kolak}, \bibinfo{person}{Ruben
  Martins}, \bibinfo{person}{Claire~Le Goues}, {and}
  \bibinfo{person}{Vincent~Josua Hellendoorn}.}
  \bibinfo{year}{2022}\natexlab{}.
\newblock \showarticletitle{Patch Generation with Language Models: Feasibility
  and Scaling Behavior}. In \bibinfo{booktitle}{\emph{Deep Learning for Code
  Workshop}}.
\newblock
\urldef\tempurl%
\url{https://openreview.net/forum?id=rHlzJh_b1-5}
\showURL{%
\tempurl}


\bibitem[Koyuncu et~al\mbox{.}(2020)]%
        {koyuncu2020fixminder}
\bibfield{author}{\bibinfo{person}{Anil Koyuncu}, \bibinfo{person}{Kui Liu},
  \bibinfo{person}{Tegawend\'{e}~F. Bissyand\'{e}}, \bibinfo{person}{Dongsun
  Kim}, \bibinfo{person}{Jacques Klein}, \bibinfo{person}{Martin Monperrus},
  {and} \bibinfo{person}{Yves Le~Traon}.} \bibinfo{year}{2020}\natexlab{}.
\newblock \showarticletitle{FixMiner: Mining Relevant Fix Patterns for
  Automated Program Repair}.
\newblock \bibinfo{journal}{\emph{Empirical Softw. Engg.}}
  \bibinfo{volume}{25}, \bibinfo{number}{3} (\bibinfo{date}{may}
  \bibinfo{year}{2020}), \bibinfo{pages}{1980–2024}.
\newblock
\showISSN{1382-3256}
\urldef\tempurl%
\url{https://doi.org/10.1007/s10664-019-09780-z}
\showDOI{\tempurl}


\bibitem[Le et~al\mbox{.}(2017)]%
        {le2017s3}
\bibfield{author}{\bibinfo{person}{Xuan-Bach~D. Le}, \bibinfo{person}{Duc-Hiep
  Chu}, \bibinfo{person}{David Lo}, \bibinfo{person}{Claire Le~Goues}, {and}
  \bibinfo{person}{Willem Visser}.} \bibinfo{year}{2017}\natexlab{}.
\newblock \showarticletitle{S3: Syntax- and Semantic-Guided Repair Synthesis
  via Programming by Examples}. In \bibinfo{booktitle}{\emph{Proceedings of the
  2017 11th Joint Meeting on Foundations of Software Engineering}} (Paderborn,
  Germany) \emph{(\bibinfo{series}{ESEC/FSE 2017})}.
  \bibinfo{publisher}{Association for Computing Machinery},
  \bibinfo{address}{New York, NY, USA}, \bibinfo{pages}{593–604}.
\newblock
\showISBNx{9781450351058}
\urldef\tempurl%
\url{https://doi.org/10.1145/3106237.3106309}
\showDOI{\tempurl}


\bibitem[Le et~al\mbox{.}(2016)]%
        {le2016hdrepair}
\bibfield{author}{\bibinfo{person}{Xuan Bach~D. Le}, \bibinfo{person}{David
  Lo}, {and} \bibinfo{person}{Claire Le~Goues}.}
  \bibinfo{year}{2016}\natexlab{}.
\newblock \showarticletitle{History Driven Program Repair}. In
  \bibinfo{booktitle}{\emph{SANER (2016)}}, Vol.~\bibinfo{volume}{1}.
  \bibinfo{pages}{213--224}.
\newblock
\urldef\tempurl%
\url{https://doi.org/10.1109/SANER.2016.76}
\showDOI{\tempurl}


\bibitem[Le~Goues et~al\mbox{.}(2012)]%
        {legoues2012genprog}
\bibfield{author}{\bibinfo{person}{Claire Le~Goues}, \bibinfo{person}{ThanhVu
  Nguyen}, \bibinfo{person}{Stephanie Forrest}, {and} \bibinfo{person}{Westley
  Weimer}.} \bibinfo{year}{2012}\natexlab{}.
\newblock \showarticletitle{GenProg: A Generic Method for Automatic Software
  Repair}.
\newblock \bibinfo{journal}{\emph{IEEE Transactions on Software Engineering}}
  \bibinfo{volume}{38}, \bibinfo{number}{1} (\bibinfo{year}{2012}),
  \bibinfo{pages}{54--72}.
\newblock
\urldef\tempurl%
\url{https://doi.org/10.1109/TSE.2011.104}
\showDOI{\tempurl}


\bibitem[Li et~al\mbox{.}(2022)]%
        {alphacode}
\bibfield{author}{\bibinfo{person}{Yujia Li}, \bibinfo{person}{David Choi},
  \bibinfo{person}{Junyoung Chung}, \bibinfo{person}{Nate Kushman},
  \bibinfo{person}{Julian Schrittwieser}, \bibinfo{person}{Rémi Leblond},
  \bibinfo{person}{Tom Eccles}, \bibinfo{person}{James Keeling},
  \bibinfo{person}{Felix Gimeno}, \bibinfo{person}{Agustin~Dal Lago},
  \bibinfo{person}{Thomas Hubert}, \bibinfo{person}{Peter Choy},
  \bibinfo{person}{Cyprien de Masson~d’Autume}, \bibinfo{person}{Igor
  Babuschkin}, \bibinfo{person}{Xinyun Chen}, \bibinfo{person}{Po-Sen Huang},
  \bibinfo{person}{Johannes Welbl}, \bibinfo{person}{Sven Gowal},
  \bibinfo{person}{Alexey Cherepanov}, \bibinfo{person}{James Molloy},
  \bibinfo{person}{Daniel~J. Mankowitz}, \bibinfo{person}{Esme~Sutherland
  Robson}, \bibinfo{person}{Pushmeet Kohli}, \bibinfo{person}{Nando de
  Freitas}, \bibinfo{person}{Koray Kavukcuoglu}, {and} \bibinfo{person}{Oriol
  Vinyals}.} \bibinfo{year}{2022}\natexlab{}.
\newblock \showarticletitle{Competition-level code generation with AlphaCode}.
\newblock \bibinfo{journal}{\emph{Science}} \bibinfo{volume}{378},
  \bibinfo{number}{6624} (\bibinfo{year}{2022}), \bibinfo{pages}{1092--1097}.
\newblock
\urldef\tempurl%
\url{https://doi.org/10.1126/science.abq1158}
\showDOI{\tempurl}
\showeprint{https://www.science.org/doi/pdf/10.1126/science.abq1158}


\bibitem[Li et~al\mbox{.}(2020)]%
        {li2020dlfix}
\bibfield{author}{\bibinfo{person}{Yi Li}, \bibinfo{person}{Shaohua Wang},
  {and} \bibinfo{person}{Tien~N. Nguyen}.} \bibinfo{year}{2020}\natexlab{}.
\newblock \showarticletitle{DLFix: Context-Based Code Transformation Learning
  for Automated Program Repair}. In \bibinfo{booktitle}{\emph{Proceedings of
  the ACM/IEEE 42nd International Conference on Software Engineering}} (Seoul,
  South Korea) \emph{(\bibinfo{series}{ICSE '20})}.
  \bibinfo{publisher}{Association for Computing Machinery},
  \bibinfo{address}{New York, NY, USA}, \bibinfo{pages}{602–614}.
\newblock
\showISBNx{9781450371216}
\urldef\tempurl%
\url{https://doi.org/10.1145/3377811.3380345}
\showDOI{\tempurl}


\bibitem[Liu et~al\mbox{.}(2023a)]%
        {liu2023your}
\bibfield{author}{\bibinfo{person}{Jiawei Liu}, \bibinfo{person}{Chunqiu~Steven
  Xia}, \bibinfo{person}{Yuyao Wang}, {and} \bibinfo{person}{Lingming Zhang}.}
  \bibinfo{year}{2023}\natexlab{a}.
\newblock \bibinfo{title}{Is Your Code Generated by ChatGPT Really Correct?
  Rigorous Evaluation of Large Language Models for Code Generation}.
\newblock
\newblock
\showeprint[arxiv]{2305.01210}~[cs.SE]


\bibitem[Liu et~al\mbox{.}(2019)]%
        {liu2019tbar}
\bibfield{author}{\bibinfo{person}{Kui Liu}, \bibinfo{person}{Anil Koyuncu},
  \bibinfo{person}{Dongsun Kim}, {and} \bibinfo{person}{Tegawend\'{e}~F.
  Bissyand\'{e}}.} \bibinfo{year}{2019}\natexlab{}.
\newblock \showarticletitle{TBar: Revisiting Template-Based Automated Program
  Repair}. In \bibinfo{booktitle}{\emph{Proceedings of the 28th ACM SIGSOFT
  International Symposium on Software Testing and Analysis}} (Beijing, China)
  \emph{(\bibinfo{series}{ISSTA 2019})}. \bibinfo{publisher}{Association for
  Computing Machinery}, \bibinfo{address}{New York, NY, USA},
  \bibinfo{pages}{31–42}.
\newblock
\showISBNx{9781450362245}
\urldef\tempurl%
\url{https://doi.org/10.1145/3293882.3330577}
\showDOI{\tempurl}


\bibitem[Liu et~al\mbox{.}(2023b)]%
        {liu2019avatar}
\bibfield{author}{\bibinfo{person}{Kui Liu}, \bibinfo{person}{Jingtang Zhang},
  \bibinfo{person}{Li Li}, \bibinfo{person}{Anil Koyuncu},
  \bibinfo{person}{Dongsun Kim}, \bibinfo{person}{Chunpeng Ge},
  \bibinfo{person}{Zhe Liu}, \bibinfo{person}{Jacques Klein}, {and}
  \bibinfo{person}{Tegawend\'{e}~F. Bissyand\'{e}}.}
  \bibinfo{year}{2023}\natexlab{b}.
\newblock \showarticletitle{Reliable Fix Patterns Inferred from Static Checkers
  for Automated Program Repair}.
\newblock \bibinfo{journal}{\emph{ACM Trans. Softw. Eng. Methodol.}}
  \bibinfo{volume}{32}, \bibinfo{number}{4}, Article \bibinfo{articleno}{96}
  (\bibinfo{date}{may} \bibinfo{year}{2023}), \bibinfo{numpages}{38}~pages.
\newblock
\showISSN{1049-331X}
\urldef\tempurl%
\url{https://doi.org/10.1145/3579637}
\showDOI{\tempurl}


\bibitem[Long and Rinard(2015)]%
        {long2015spr}
\bibfield{author}{\bibinfo{person}{Fan Long} {and} \bibinfo{person}{Martin
  Rinard}.} \bibinfo{year}{2015}\natexlab{}.
\newblock \showarticletitle{Staged Program Repair with Condition Synthesis}. In
  \bibinfo{booktitle}{\emph{Proceedings of the 2015 10th Joint Meeting on
  Foundations of Software Engineering}} (Bergamo, Italy)
  \emph{(\bibinfo{series}{ESEC/FSE 2015})}. \bibinfo{publisher}{Association for
  Computing Machinery}, \bibinfo{address}{New York, NY, USA},
  \bibinfo{pages}{166–178}.
\newblock
\showISBNx{9781450336758}
\urldef\tempurl%
\url{https://doi.org/10.1145/2786805.2786811}
\showDOI{\tempurl}


\bibitem[Lutellier et~al\mbox{.}(2020)]%
        {lutellier2020coconut}
\bibfield{author}{\bibinfo{person}{Thibaud Lutellier},
  \bibinfo{person}{Hung~Viet Pham}, \bibinfo{person}{Lawrence Pang},
  \bibinfo{person}{Yitong Li}, \bibinfo{person}{Moshi Wei}, {and}
  \bibinfo{person}{Lin Tan}.} \bibinfo{year}{2020}\natexlab{}.
\newblock \showarticletitle{CoCoNuT: Combining Context-Aware Neural Translation
  Models Using Ensemble for Program Repair}. In
  \bibinfo{booktitle}{\emph{Proceedings of the 29th ACM SIGSOFT International
  Symposium on Software Testing and Analysis}} (Virtual Event, USA)
  \emph{(\bibinfo{series}{ISSTA 2020})}. \bibinfo{publisher}{Association for
  Computing Machinery}, \bibinfo{address}{New York, NY, USA},
  \bibinfo{pages}{101–114}.
\newblock
\showISBNx{9781450380089}
\urldef\tempurl%
\url{https://doi.org/10.1145/3395363.3397369}
\showDOI{\tempurl}


\bibitem[Martinez et~al\mbox{.}(2017)]%
        {martinez2015automatic}
\bibfield{author}{\bibinfo{person}{Matias Martinez}, \bibinfo{person}{Thomas
  Durieux}, \bibinfo{person}{Romain Sommerard}, \bibinfo{person}{Jifeng Xuan},
  {and} \bibinfo{person}{Martin Monperrus}.} \bibinfo{year}{2017}\natexlab{}.
\newblock \showarticletitle{Automatic Repair of Real Bugs in Java: A
  Large-Scale Experiment on the Defects4j Dataset}.
\newblock \bibinfo{journal}{\emph{Empirical Softw. Engg.}}
  \bibinfo{volume}{22}, \bibinfo{number}{4} (\bibinfo{date}{aug}
  \bibinfo{year}{2017}), \bibinfo{pages}{1936–1964}.
\newblock
\showISSN{1382-3256}
\urldef\tempurl%
\url{https://doi.org/10.1007/s10664-016-9470-4}
\showDOI{\tempurl}


\bibitem[Martinez and Monperrus(2016)]%
        {martinez2016astor}
\bibfield{author}{\bibinfo{person}{Matias Martinez} {and}
  \bibinfo{person}{Martin Monperrus}.} \bibinfo{year}{2016}\natexlab{}.
\newblock \showarticletitle{ASTOR: A Program Repair Library for Java (Demo)}.
  In \bibinfo{booktitle}{\emph{Proceedings of the 25th International Symposium
  on Software Testing and Analysis}} (Saarbr\"{u}cken, Germany)
  \emph{(\bibinfo{series}{ISSTA 2016})}. \bibinfo{publisher}{Association for
  Computing Machinery}, \bibinfo{address}{New York, NY, USA},
  \bibinfo{pages}{441–444}.
\newblock
\showISBNx{9781450343909}
\urldef\tempurl%
\url{https://doi.org/10.1145/2931037.2948705}
\showDOI{\tempurl}


\bibitem[Mechtaev et~al\mbox{.}(2016)]%
        {mechtaev2016angelix}
\bibfield{author}{\bibinfo{person}{Sergey Mechtaev}, \bibinfo{person}{Jooyong
  Yi}, {and} \bibinfo{person}{Abhik Roychoudhury}.}
  \bibinfo{year}{2016}\natexlab{}.
\newblock \showarticletitle{Angelix: Scalable Multiline Program Patch Synthesis
  via Symbolic Analysis}. In \bibinfo{booktitle}{\emph{Proceedings of the 38th
  International Conference on Software Engineering}} (Austin, Texas)
  \emph{(\bibinfo{series}{ICSE '16})}. \bibinfo{publisher}{Association for
  Computing Machinery}, \bibinfo{address}{New York, NY, USA},
  \bibinfo{pages}{691–701}.
\newblock
\showISBNx{9781450339001}
\urldef\tempurl%
\url{https://doi.org/10.1145/2884781.2884807}
\showDOI{\tempurl}


\bibitem[Microsoft(2023a)]%
        {lsp}
Microsoft \bibinfo{year}{2023}\natexlab{a}.
\newblock \bibinfo{title}{Language Server Protocol}.
\newblock
\newblock
\newblock
\shownote{\url{https://microsoft.github.io/language-server-protocol}}.


\bibitem[Microsoft(2023b)]%
        {typescript}
Microsoft \bibinfo{year}{2023}\natexlab{b}.
\newblock \bibinfo{title}{TypeScript}.
\newblock
\newblock
\newblock
\shownote{\url{https://www.typescriptlang.org}}.


\bibitem[Nijkamp et~al\mbox{.}(2023)]%
        {nijkamp2023codegen2}
\bibfield{author}{\bibinfo{person}{Erik Nijkamp}, \bibinfo{person}{Hiroaki
  Hayashi}, \bibinfo{person}{Caiming Xiong}, \bibinfo{person}{Silvio Savarese},
  {and} \bibinfo{person}{Yingbo Zhou}.} \bibinfo{year}{2023}\natexlab{}.
\newblock \bibinfo{title}{CodeGen2: Lessons for Training LLMs on Programming
  and Natural Languages}.
\newblock
\newblock
\showeprint[arxiv]{2305.02309}~[cs.LG]


\bibitem[Nijkamp et~al\mbox{.}(2022)]%
        {Nijkamp2022CG}
\bibfield{author}{\bibinfo{person}{Erik Nijkamp}, \bibinfo{person}{Bo Pang},
  \bibinfo{person}{Hiroaki Hayashi}, \bibinfo{person}{Lifu Tu},
  \bibinfo{person}{Huan Wang}, \bibinfo{person}{Yingbo Zhou},
  \bibinfo{person}{Silvio Savarese}, {and} \bibinfo{person}{Caiming Xiong}.}
  \bibinfo{year}{2022}\natexlab{}.
\newblock \bibinfo{title}{CodeGen: An Open Large Language Model for Code with
  Multi-Turn Program Synthesis}.
\newblock
\newblock
\newblock
\shownote{arXiv:2203.13474}.


\bibitem[Poesia et~al\mbox{.}(2022)]%
        {poesia2022synchromesh}
\bibfield{author}{\bibinfo{person}{Gabriel Poesia}, \bibinfo{person}{Alex
  Polozov}, \bibinfo{person}{Vu Le}, \bibinfo{person}{Ashish Tiwari},
  \bibinfo{person}{Gustavo Soares}, \bibinfo{person}{Christopher Meek}, {and}
  \bibinfo{person}{Sumit Gulwani}.} \bibinfo{year}{2022}\natexlab{}.
\newblock \showarticletitle{Synchromesh: Reliable Code Generation from
  Pre-trained Language Models}. In \bibinfo{booktitle}{\emph{International
  Conference on Learning Representations}}.
\newblock
\urldef\tempurl%
\url{https://openreview.net/forum?id=KmtVD97J43e}
\showURL{%
\tempurl}


\bibitem[Prenner et~al\mbox{.}(2022)]%
        {prenner2021codexws}
\bibfield{author}{\bibinfo{person}{Julian~Aron Prenner}, \bibinfo{person}{Hlib
  Babii}, {and} \bibinfo{person}{Romain Robbes}.}
  \bibinfo{year}{2022}\natexlab{}.
\newblock \showarticletitle{Can OpenAI's Codex Fix Bugs? An Evaluation on
  QuixBugs}. In \bibinfo{booktitle}{\emph{APR '22}} (Pittsburgh, Pennsylvania).
  \bibinfo{publisher}{Association for Computing Machinery},
  \bibinfo{address}{New York, NY, USA}, \bibinfo{pages}{69–75}.
\newblock
\showISBNx{9781450392853}
\urldef\tempurl%
\url{https://doi.org/10.1145/3524459.3527351}
\showDOI{\tempurl}


\bibitem[Python(2023)]%
        {python}
Python \bibinfo{year}{2023}\natexlab{}.
\newblock \bibinfo{title}{Type Hints in Python}.
\newblock
\newblock
\newblock
\shownote{\url{https://peps.python.org/pep-0484/}}.


\bibitem[Sennrich et~al\mbox{.}(2016)]%
        {sennrich2015neural}
\bibfield{author}{\bibinfo{person}{Rico Sennrich}, \bibinfo{person}{Barry
  Haddow}, {and} \bibinfo{person}{Alexandra Birch}.}
  \bibinfo{year}{2016}\natexlab{}.
\newblock \showarticletitle{Neural Machine Translation of Rare Words with
  Subword Units}. In \bibinfo{booktitle}{\emph{Proceedings of the 54th Annual
  Meeting of the Association for Computational Linguistics (Volume 1: Long
  Papers)}}. \bibinfo{publisher}{Association for Computational Linguistics},
  \bibinfo{address}{Berlin, Germany}, \bibinfo{pages}{1715--1725}.
\newblock
\urldef\tempurl%
\url{https://doi.org/10.18653/v1/P16-1162}
\showDOI{\tempurl}


\bibitem[Sobania et~al\mbox{.}(2023)]%
        {sobania2023analysis}
\bibfield{author}{\bibinfo{person}{Dominik Sobania}, \bibinfo{person}{Martin
  Briesch}, \bibinfo{person}{Carol Hanna}, {and} \bibinfo{person}{Justyna
  Petke}.} \bibinfo{year}{2023}\natexlab{}.
\newblock \bibinfo{title}{An Analysis of the Automatic Bug Fixing Performance
  of ChatGPT}.
\newblock
\newblock
\showeprint[arxiv]{2301.08653}~[cs.SE]


\bibitem[Sutskever et~al\mbox{.}(2014)]%
        {sutskever2014mt}
\bibfield{author}{\bibinfo{person}{Ilya Sutskever}, \bibinfo{person}{Oriol
  Vinyals}, {and} \bibinfo{person}{Quoc~V Le}.}
  \bibinfo{year}{2014}\natexlab{}.
\newblock \showarticletitle{Sequence to Sequence Learning with Neural
  Networks}. In \bibinfo{booktitle}{\emph{Advances in Neural Information
  Processing Systems}}, \bibfield{editor}{\bibinfo{person}{Z.~Ghahramani},
  \bibinfo{person}{M.~Welling}, \bibinfo{person}{C.~Cortes},
  \bibinfo{person}{N.~Lawrence}, {and} \bibinfo{person}{K.Q. Weinberger}}
  (Eds.), Vol.~\bibinfo{volume}{27}. \bibinfo{publisher}{Curran Associates,
  Inc.}
\newblock
\urldef\tempurl%
\url{https://proceedings.neurips.cc/paper_files/paper/2014/file/a14ac55a4f27472c5d894ec1c3c743d2-Paper.pdf}
\showURL{%
\tempurl}


\bibitem[Tufano et~al\mbox{.}(2019)]%
        {tufano2018empstudy}
\bibfield{author}{\bibinfo{person}{Michele Tufano}, \bibinfo{person}{Cody
  Watson}, \bibinfo{person}{Gabriele Bavota}, \bibinfo{person}{Massimiliano~Di
  Penta}, \bibinfo{person}{Martin White}, {and} \bibinfo{person}{Denys
  Poshyvanyk}.} \bibinfo{year}{2019}\natexlab{}.
\newblock \showarticletitle{An Empirical Study on Learning Bug-Fixing Patches
  in the Wild via Neural Machine Translation}.
\newblock \bibinfo{journal}{\emph{ACM Trans. Softw. Eng. Methodol.}}
  \bibinfo{volume}{28}, \bibinfo{number}{4}, Article \bibinfo{articleno}{19}
  (\bibinfo{date}{sep} \bibinfo{year}{2019}), \bibinfo{numpages}{29}~pages.
\newblock
\showISSN{1049-331X}
\urldef\tempurl%
\url{https://doi.org/10.1145/3340544}
\showDOI{\tempurl}


\bibitem[Vaswani et~al\mbox{.}(2017)]%
        {vaswani2017attention}
\bibfield{author}{\bibinfo{person}{Ashish Vaswani}, \bibinfo{person}{Noam
  Shazeer}, \bibinfo{person}{Niki Parmar}, \bibinfo{person}{Jakob Uszkoreit},
  \bibinfo{person}{Llion Jones}, \bibinfo{person}{Aidan~N. Gomez},
  \bibinfo{person}{\L{}ukasz Kaiser}, {and} \bibinfo{person}{Illia
  Polosukhin}.} \bibinfo{year}{2017}\natexlab{}.
\newblock \showarticletitle{Attention is All You Need}. In
  \bibinfo{booktitle}{\emph{Proceedings of the 31st International Conference on
  Neural Information Processing Systems}} (Long Beach, California, USA)
  \emph{(\bibinfo{series}{NIPS'17})}. \bibinfo{publisher}{Curran Associates
  Inc.}, \bibinfo{address}{Red Hook, NY, USA}, \bibinfo{pages}{6000–6010}.
\newblock
\showISBNx{9781510860964}


\bibitem[Wang et~al\mbox{.}(2023)]%
        {wang2023codet5}
\bibfield{author}{\bibinfo{person}{Yue Wang}, \bibinfo{person}{Hung Le},
  \bibinfo{person}{Akhilesh~Deepak Gotmare}, \bibinfo{person}{Nghi D.~Q. Bui},
  \bibinfo{person}{Junnan Li}, {and} \bibinfo{person}{Steven C.~H. Hoi}.}
  \bibinfo{year}{2023}\natexlab{}.
\newblock \bibinfo{title}{CodeT5+: Open Code Large Language Models for Code
  Understanding and Generation}.
\newblock
\newblock
\showeprint[arxiv]{2305.07922}~[cs.CL]


\bibitem[Wang et~al\mbox{.}(2021)]%
        {wang2021codet5}
\bibfield{author}{\bibinfo{person}{Yue Wang}, \bibinfo{person}{Weishi Wang},
  \bibinfo{person}{Shafiq Joty}, {and} \bibinfo{person}{Steven~C.H. Hoi}.}
  \bibinfo{year}{2021}\natexlab{}.
\newblock \showarticletitle{{C}ode{T}5: Identifier-aware Unified Pre-trained
  Encoder-Decoder Models for Code Understanding and Generation}. In
  \bibinfo{booktitle}{\emph{Proceedings of the 2021 Conference on Empirical
  Methods in Natural Language Processing}}. \bibinfo{publisher}{Association for
  Computational Linguistics}, \bibinfo{address}{Online and Punta Cana,
  Dominican Republic}, \bibinfo{pages}{8696--8708}.
\newblock
\urldef\tempurl%
\url{https://doi.org/10.18653/v1/2021.emnlp-main.685}
\showDOI{\tempurl}


\bibitem[Wei et~al\mbox{.}(2023)]%
        {correctPatchandDataset}
\bibfield{author}{\bibinfo{person}{Yuxiang Wei},
  \bibinfo{person}{Chunqiu~Steven Xia}, {and} \bibinfo{person}{Lingming
  Zhang}.} \bibinfo{year}{2023}\natexlab{}.
\newblock \bibinfo{title}{{ESEC/FSE'23 Artifact for "Copiloting the Copilots:
  Fusing Large Language Models with Completion Engines for Automated Program
  Repair"}}.
\newblock
\newblock
\urldef\tempurl%
\url{https://doi.org/10.5281/zenodo.8281250}
\showDOI{\tempurl}


\bibitem[Wen et~al\mbox{.}(2018)]%
        {wen2018capgen}
\bibfield{author}{\bibinfo{person}{Ming Wen}, \bibinfo{person}{Junjie Chen},
  \bibinfo{person}{Rongxin Wu}, \bibinfo{person}{Dan Hao}, {and}
  \bibinfo{person}{Shing-Chi Cheung}.} \bibinfo{year}{2018}\natexlab{}.
\newblock \showarticletitle{Context-Aware Patch Generation for Better Automated
  Program Repair}. In \bibinfo{booktitle}{\emph{Proceedings of the 40th
  International Conference on Software Engineering}} (Gothenburg, Sweden)
  \emph{(\bibinfo{series}{ICSE '18})}. \bibinfo{publisher}{Association for
  Computing Machinery}, \bibinfo{address}{New York, NY, USA},
  \bibinfo{pages}{1–11}.
\newblock
\showISBNx{9781450356381}
\urldef\tempurl%
\url{https://doi.org/10.1145/3180155.3180233}
\showDOI{\tempurl}


\bibitem[Xia et~al\mbox{.}(2023a)]%
        {xia2023revisiting}
\bibfield{author}{\bibinfo{person}{Chunqiu~Steven Xia}, \bibinfo{person}{Yifeng
  Ding}, {and} \bibinfo{person}{Lingming Zhang}.}
  \bibinfo{year}{2023}\natexlab{a}.
\newblock \bibinfo{title}{Revisiting the Plastic Surgery Hypothesis via Large
  Language Models}.
\newblock
\newblock
\showeprint[arxiv]{2303.10494}~[cs.SE]


\bibitem[Xia et~al\mbox{.}(2023b)]%
        {fuzz4all}
\bibfield{author}{\bibinfo{person}{Chunqiu~Steven Xia}, \bibinfo{person}{Matteo
  Paltenghi}, \bibinfo{person}{Jia~Le Tian}, \bibinfo{person}{Michael Pradel},
  {and} \bibinfo{person}{Lingming Zhang}.} \bibinfo{year}{2023}\natexlab{b}.
\newblock \bibinfo{title}{Universal Fuzzing via Large Language Models}.
\newblock
\newblock
\showeprint[arxiv]{2308.04748}~[cs.SE]


\bibitem[Xia et~al\mbox{.}(2023c)]%
        {xia2022repairstudy}
\bibfield{author}{\bibinfo{person}{Chunqiu~Steven Xia},
  \bibinfo{person}{Yuxiang Wei}, {and} \bibinfo{person}{Lingming Zhang}.}
  \bibinfo{year}{2023}\natexlab{c}.
\newblock \showarticletitle{Automated Program Repair in the Era of Large
  Pre-Trained Language Models}. In \bibinfo{booktitle}{\emph{Proceedings of the
  45th International Conference on Software Engineering}} (Melbourne, Victoria,
  Australia) \emph{(\bibinfo{series}{ICSE '23})}. \bibinfo{publisher}{IEEE
  Press}, \bibinfo{pages}{1482–1494}.
\newblock
\showISBNx{9781665457019}
\urldef\tempurl%
\url{https://doi.org/10.1109/ICSE48619.2023.00129}
\showDOI{\tempurl}


\bibitem[Xia and Zhang(2022)]%
        {xia2022alpharepair}
\bibfield{author}{\bibinfo{person}{Chunqiu~Steven Xia} {and}
  \bibinfo{person}{Lingming Zhang}.} \bibinfo{year}{2022}\natexlab{}.
\newblock \showarticletitle{Less Training, More Repairing Please: Revisiting
  Automated Program Repair via Zero-Shot Learning}. In
  \bibinfo{booktitle}{\emph{Proceedings of the 30th ACM Joint European Software
  Engineering Conference and Symposium on the Foundations of Software
  Engineering}} (Singapore, Singapore) \emph{(\bibinfo{series}{ESEC/FSE
  2022})}. \bibinfo{publisher}{Association for Computing Machinery},
  \bibinfo{address}{New York, NY, USA}, \bibinfo{pages}{959–971}.
\newblock
\showISBNx{9781450394130}
\urldef\tempurl%
\url{https://doi.org/10.1145/3540250.3549101}
\showDOI{\tempurl}


\bibitem[Xia and Zhang(2023a)]%
        {xia2023conversational}
\bibfield{author}{\bibinfo{person}{Chunqiu~Steven Xia} {and}
  \bibinfo{person}{Lingming Zhang}.} \bibinfo{year}{2023}\natexlab{a}.
\newblock \showarticletitle{Conversational Automated Program Repair}.
\newblock \bibinfo{journal}{\emph{CoRR}}  \bibinfo{volume}{abs/2301.13246}
  (\bibinfo{year}{2023}).
\newblock
\urldef\tempurl%
\url{https://doi.org/10.48550/ARXIV.2301.13246}
\showDOI{\tempurl}
\showeprint[arXiv]{2301.13246}


\bibitem[Xia and Zhang(2023b)]%
        {xia2023conversation}
\bibfield{author}{\bibinfo{person}{Chunqiu~Steven Xia} {and}
  \bibinfo{person}{Lingming Zhang}.} \bibinfo{year}{2023}\natexlab{b}.
\newblock \bibinfo{title}{Keep the Conversation Going: Fixing 162 out of 337
  bugs for \$0.42 each using ChatGPT}.
\newblock
\newblock
\showeprint[arxiv]{2304.00385}~[cs.SE]


\bibitem[Xu et~al\mbox{.}(2022)]%
        {xu2022systematic}
\bibfield{author}{\bibinfo{person}{Frank~F. Xu}, \bibinfo{person}{Uri Alon},
  \bibinfo{person}{Graham Neubig}, {and} \bibinfo{person}{Vincent~Josua
  Hellendoorn}.} \bibinfo{year}{2022}\natexlab{}.
\newblock \showarticletitle{A Systematic Evaluation of Large Language Models of
  Code}. In \bibinfo{booktitle}{\emph{Proceedings of the 6th ACM SIGPLAN
  International Symposium on Machine Programming}} (San Diego, CA, USA)
  \emph{(\bibinfo{series}{MAPS 2022})}. \bibinfo{publisher}{Association for
  Computing Machinery}, \bibinfo{address}{New York, NY, USA},
  \bibinfo{pages}{1–10}.
\newblock
\showISBNx{9781450392730}
\urldef\tempurl%
\url{https://doi.org/10.1145/3520312.3534862}
\showDOI{\tempurl}


\bibitem[Yang et~al\mbox{.}(2019)]%
        {yang2020xlnet}
\bibfield{author}{\bibinfo{person}{Zhilin Yang}, \bibinfo{person}{Zihang Dai},
  \bibinfo{person}{Yiming Yang}, \bibinfo{person}{Jaime Carbonell},
  \bibinfo{person}{Ruslan Salakhutdinov}, {and} \bibinfo{person}{Quoc~V. Le}.}
  \bibinfo{year}{2019}\natexlab{}.
\newblock \bibinfo{booktitle}{\emph{XLNet: Generalized Autoregressive
  Pretraining for Language Understanding}}.
\newblock \bibinfo{publisher}{Curran Associates Inc.}, \bibinfo{address}{Red
  Hook, NY, USA}.
\newblock


\bibitem[Ye et~al\mbox{.}(2022)]%
        {ye2022rewardrepair}
\bibfield{author}{\bibinfo{person}{He Ye}, \bibinfo{person}{Matias Martinez},
  {and} \bibinfo{person}{Martin Monperrus}.} \bibinfo{year}{2022}\natexlab{}.
\newblock \showarticletitle{Neural Program Repair with Execution-Based
  Backpropagation}. In \bibinfo{booktitle}{\emph{Proceedings of the 44th
  International Conference on Software Engineering}} (Pittsburgh, Pennsylvania)
  \emph{(\bibinfo{series}{ICSE '22})}. \bibinfo{publisher}{Association for
  Computing Machinery}, \bibinfo{address}{New York, NY, USA},
  \bibinfo{pages}{1506–1518}.
\newblock
\showISBNx{9781450392211}
\urldef\tempurl%
\url{https://doi.org/10.1145/3510003.3510222}
\showDOI{\tempurl}


\bibitem[Zhang et~al\mbox{.}(2023)]%
        {zhang2023repocoder}
\bibfield{author}{\bibinfo{person}{Fengji Zhang}, \bibinfo{person}{Bei Chen},
  \bibinfo{person}{Yue Zhang}, \bibinfo{person}{Jacky Keung},
  \bibinfo{person}{Jin Liu}, \bibinfo{person}{Daoguang Zan},
  \bibinfo{person}{Yi Mao}, \bibinfo{person}{Jian-Guang Lou}, {and}
  \bibinfo{person}{Weizhu Chen}.} \bibinfo{year}{2023}\natexlab{}.
\newblock \bibinfo{title}{RepoCoder: Repository-Level Code Completion Through
  Iterative Retrieval and Generation}.
\newblock
\newblock
\showeprint[arxiv]{2303.12570}~[cs.CL]


\bibitem[Zhu et~al\mbox{.}(2021)]%
        {zhu2021recoder}
\bibfield{author}{\bibinfo{person}{Qihao Zhu}, \bibinfo{person}{Zeyu Sun},
  \bibinfo{person}{Yuan-an Xiao}, \bibinfo{person}{Wenjie Zhang},
  \bibinfo{person}{Kang Yuan}, \bibinfo{person}{Yingfei Xiong}, {and}
  \bibinfo{person}{Lu Zhang}.} \bibinfo{year}{2021}\natexlab{}.
\newblock \showarticletitle{A Syntax-Guided Edit Decoder for Neural Program
  Repair}. In \bibinfo{booktitle}{\emph{ESEC/FSE 2021}} (Athens, Greece).
  \bibinfo{publisher}{Association for Computing Machinery},
  \bibinfo{address}{New York, NY, USA}, \bibinfo{pages}{341–353}.
\newblock
\showISBNx{9781450385626}
\urldef\tempurl%
\url{https://doi.org/10.1145/3468264.3468544}
\showDOI{\tempurl}


\end{thebibliography}

\end{document}